\documentclass[english,A4paper, margin=1in,11pt]{article}
\usepackage{geometry}
\usepackage{todonotes}
\usepackage{lmodern}

\usepackage[T1]{fontenc}
\usepackage[latin9]{inputenc}
\usepackage{babel}
\usepackage{textcomp}
\usepackage{amsmath} 
\usepackage{amsthm} 
\usepackage{amssymb} 
\usepackage{dsfont} 

\usepackage{graphicx}
\usepackage{epstopdf}
\usepackage{todonotes}
\usepackage{comment}
\usepackage[unicode=true,pdfusetitle,
bookmarks=true,bookmarksnumbered=false,bookmarksopen=false,
breaklinks=false,pdfborder={0 0 0},pdfborderstyle={},backref=false,colorlinks=false]
{hyperref}
\usepackage{breakurl}

\makeatletter 
\numberwithin{equation}{section}
\theoremstyle{plain} 
\newtheorem{thm}{\protect\theoremname}[section]
\theoremstyle{definition}
\newtheorem{definition}[thm]{\protect\definitionname}

\theoremstyle{remark}
\newtheorem{remark}[thm]{\protect\remarkname}
\theoremstyle{plain}
\newtheorem{lemma}[thm]{\protect\lemmaname}
\theoremstyle{plain}
\newtheorem{proposition}[thm]{\protect\propositionname}
\theoremstyle{definition}
\newtheorem{example}[thm]{\protect\examplename}
\theoremstyle{plain}
\newtheorem{corollary}[thm]{\protect\corollaryname}

\usepackage[round]{natbib}

\makeatother
\providecommand{\corollaryname}{Corollary}
\providecommand{\definitionname}{Definition}
\providecommand{\examplename}{Example}
\providecommand{\lemmaname}{Lemma}
\providecommand{\propositionname}{Proposition}
\providecommand{\remarkname}{Remark}
\providecommand{\theoremname}{Theorem}

\newtheorem{assumption}{Condition}
\newtheorem{problem}{Problem}

\renewcommand{\geq}{\geqslant}
\renewcommand{\leq}{\leqslant}

\usepackage{setspace}
\begin{document}
	\global\long\def\Pa{\mathbb{P}^{\ast}}
	
	\global\long\def\P{\mathbb{P}}
	
	\global\long\def\R{\mathbb{R}}
	
	\global\long\def\Q{\mathbb{Q}}
	
	\global\long\def\PinP{\mathbb{P}\in\mathcal{P}}
	
	\global\long\def\Xa{X^{\ast}}
	
	\global\long\def\ella{\ell^{\ast}}
	
	\global\long\def\E{E}

	\title{\bf Cost-efficient Payoffs under Model Ambiguity}
	\author{Carole Bernard\thanks{\underline{Corresponding author:} 
			Carole Bernard, Department of Accounting, Law and Finance at Grenoble
			Ecole de Management (email: \texttt{carole.bernard@grenoble-em.com}) and Department of Economics and Political Sciences at Vrije
			Universiteit Brussel (VUB). },\\
		Gero Junike\thanks{Institute of Mathematics, Carl von Ossietzky University, 26111 Oldenburg, Germany (e-mail: \texttt{gero.junike@uol.de}). }\\
		Thibaut Lux\thanks{Helvetia Insurance Group (e-mail: \texttt{thibaut.lux@helvetia.ch}).}\\
		Steven Vanduffel\thanks{ Steven Vanduffel, Department of Economics and Political Sciences at Vrije
			Universiteit Brussel (VUB). (e-mail: \texttt{steven.vanduffel@vub.ac.be}). } \thanks{We thank Ruodu Wang for constructive comments on an earlier version. The views expressed in this manuscript do not represent the opinions of Helvetia Insurance Group. This work was supported by the FWO  projects G015320N and S006721N. }
	}
	\date{\today} 
	\maketitle

\begin{abstract}
\cite{dybvigJoB,dybvigRFS} solves in a complete market setting the problem of finding a payoff that is cheapest possible in reaching a given target distribution  (``cost-efficient payoff''). In the presence of ambiguity, the distribution of a payoff is, however, no longer known with certainty. We study the problem of finding the cheapest possible payoff whose  worst-case distribution stochastically dominates a given target distribution (``robust cost-efficient payoff'') and determine solutions under certain conditions. We study the link between ``robust cost-efficiency'' and the maxmin expected utility setting of Gilboa and Schmeidler, as well as more generally with robust preferences in a possibly non-expected utility setting. Specifically, we show that solutions to maxmin robust expected utility are necessarily robust cost-efficient. We illustrate our study with examples involving uncertainty both on the drift and on the volatility of the
risky asset. 
\\

\end{abstract}
\newpage

\section{Introduction}

In a (complete) market without ambiguity, \cite{dybvigJoB,dybvigRFS}
characterizes optimal payoffs for agents having law-invariant increasing
preferences (e.g., expected utility maximizers). His result is based on the observation that any optimal
payoff $X$ must be cost-efficient in the sense that there cannot exist
another payoff with the same probability distribution that is strictly
cheaper than $X$. He then derives, for a given target distribution
of terminal wealth, the payoff that achieves this target distribution
at the lowest possible cost (cost-efficient payoff). 
Optimal portfolios are thus driven by distributional constraints rather than appearing as a solution to some optimal expected utility problem.  
In this regard, \cite{brennan1981optimal} note that 
\textit{``from a practical point of view it may well prove easier for the investor to choose directly his optimal payoff function than it would be for him to communicate his utility function to a portfolio manager.''} 
\cite{sharpe2000distribution} and  \cite{goldstein2008choosing} introduce a tool called the distribution builder, which makes it possible for investors to analyze distributions of terminal wealth and to choose their preferred one among
alternatives with equal cost; see also  
\cite{sharpe2011investors} and \cite{monin2014dynamic}.  Moreover, \cite{goldstein2008choosing} argue that such a tool makes it possible to better elicit investor's preferences. 

Our main objective is to extend Dybvig's results when there is 
uncertainty on the real-world probability measure. Uncertainty has become a prime issue in many academic domains, from
economics to environmental science and psychology. 
Model ambiguity refers to random phenomena or outcomes whose probabilities
are themselves unknown. For instance, the random outcome of a coin
toss is subject to model uncertainty when the probability of the coin
showing either a head or a tail is not or is at most partially known.
This notion of model ambiguity goes back to \cite{knight1921risk}
and is therefore commonly referred to as Knightian uncertainty.

In the presence of ambiguity, the probability distribution of a payoff is not anymore determined. Thus,  looking for a minimum cost payoff with a given probability distribution is no longer possible. However, investors may still determine a desired distribution function that they would like to achieve ``at least''.  In this paper, we look for  a minimum cost payoff that dominates  a target distribution for a chosen stochastic integral order under any plausible real-world probability distribution.  Our contributions are three-fold. First, we solve this problem explicitly for a general stochastic ordering under certain assumptions. Solutions to this problem are called ``robust cost-efficient.'' Second, we draw connections between such a minimum cost payoff and the problem of finding an optimal portfolio under ambiguity for general sets of robust preferences. Third, we present a number of examples, including one on the robust portfolio choice in the presence of volatility uncertainty.


Our results generalize the results on cost-efficiency
given in \cite{dybvigJoB,dybvigRFS}, \cite{cox2000dynamic} and \cite{bernard2013explicit,bernard2015rationalizing}. When there is no ambiguity on the real-world probability, the robust cost efficient payoffs coincide with the cost-efficient payoffs studied in the literature.  
To derive our results, we build on  the so-called quantile approach to solve the optimization of a law invariant increasing functional, e.g., \cite{schied2004neyman}, \cite{carlier2006law, carlier2008two, carlier2011optimal}, \cite{JZ08},  \cite{he2011portfolio,he2011portfoliob}, \cite{bernard2013explicit}, \cite{xu2016note} and \cite{ruschendorf2020construction}.

Specifically, we consider a static setting but we have incomplete markets and
are thus able to address uncertainty about volatility. We show that, under certain conditions, the solution to a general robust portfolio
maximization problem is equal to the solution of a classical portfolio
maximization problem under a \emph{least favorable measure} $\mathbb{P}^{\ast}$
with respect to some stochastic ordering. This was already shown by \cite{schied2005optimal} for the case of robust expected utility theory and using first order stochastic dominance ordering; here, however, we show that these results extend to the case of more general preferences. Specifically, we focus on the case of first order and second order stochastic dominance. 

Furthermore, we show that there is a natural correspondence between
optimal portfolios in the maxmin utility setting of \cite{gilboa1989maxmin}, with a concave increasing utility, and robust cost-efficient payoffs:
for any robust cost-efficient payoff $X^{\ast}$, there is a 
utility function such that $X^{\ast}$ solves the maxmin expected
utility maximization problem. We further show that the solution to
a robust maximization problem with respect to a general family of
preferences is cost-efficient. This result implies that instead of
solving a robust maximization problem with respect to a general family
of preferences one could solve an expected utility maximization problem
under the single measure $\mathbb{P}^{\ast}$ for a suitable concave
utility function.

	The literature on optimal payoff choice under ambiguity includes the seminal setting of  \cite{gilboa1989maxmin} that is, the so-called  ``maxmin expected utility,'' which was later referred to as robust utility functional by \cite{schied2009robust}.  Specifically, these authors
characterize preferences that have a robust utility numerical
representation $\underset{\mathbb{P}\in\mathcal{P}}{\min}\,E_{\mathbb{P}}[u(X)]$ for some set of probabilities $\mathcal{P}$.
\cite{gundel2005robust} provides a dual characterization of the
solution for robust utility maximization in both a complete and an
incomplete market model. \cite{klibanoff2005smooth} distinguish between
subjective beliefs, e.g., the definition of the set of possible or
plausible subjective probability measures, and ambiguity attitude,
i.e., a characterization of the agent's behavior toward ambiguity.
Based on \cite{klibanoff2005smooth}, \cite{gollier2011portfolio}
analyzes the effect of ambiguity aversion on the demand for the uncertain
asset in a portfolio choice problem.


\cite{schied2005optimal} solves the maximization problem of maxmin
expected utility of \cite{gilboa1989maxmin} in a general complete market model with dynamic trading,
provided there is a least favorable measure with respect
to FSD ordering. Specifically, he finds that the optimum for the maxmin
utility setting of \cite{gilboa1989maxmin} can be derived in the
standard expected utility setting under the least favorable measure.
\cite{schied2005optimal} works with a complete market model and mainly
in a static setting; dynamics only come into play when the martingale
method is applied to the static solutions. A survey on robust preferences and robust portfolio choice can be found in \cite{schied2009robust}.

The paper is organized as follows. The robust cost-efficiency problem is described in Section~\ref{subsec:Market-setting}. In Section~\ref{sec:Robust-cost-efficiency},
we  solve the robust cost-efficiency problem, and we
include two examples in a log-normal market with uncertainty
on the drift and the volatility along with another example in a L\'evy
market in which the physical measure is obtained by the Esscher transform.
In Section~\ref{sec:Robust-Portfolio-Selection}, we develop the correspondence
between robust cost-efficient payoffs and strategies that solve a
robust optimal portfolio problem, including the maxmin utility setting
of \cite{gilboa1989maxmin} as a special case. In Section~\ref{sec:Rationalizing-robust-cost-effici},
we show that the solution to a general robust optimal portfolio problem
can also be obtained as a solution to the maximization of the maxmin
utility setting of \cite{gilboa1989maxmin} for a well-chosen concave
utility function. Section~\ref{sec:Summary} concludes.

%

\section{Problem statement \label{subsec:Market-setting}}

We assume a static market setting in which trading only takes place today
and at the end of the planning horizon $T>0$. There is a bank account earning the continuously
compounded risk-free interest rate $r\in\mathbb{R}$. Let $\mathbb{R}_+=[0,\infty)$. Let $S_{T}:\Omega\rightarrow\mathbb{R}_+$
represent the random value of a risky asset at maturity. We denote by $S_0 >0$ its current value and by $\mathcal{F}$ the $\sigma$-algebra $S_T$ generates. 
Let $\mathcal{P}$ be a set of equivalent real-world probability measures on $(\Omega,\mathcal{F})$. The set
$\mathcal{P}$ can be thought of as a collection of probability measures that the investor deems plausible for the market. We define the set of \emph{payoffs} 
$\mathcal{X}=\left\{ g(S_{T}),\,g:\mathbb{R}_+\to\mathbb{R}_+\text{ is measurable and }  E_{\mathbb{Q}}[|g(S_{T})|]<\infty    \right\},$
where $\mathbb{Q}$ is a martingale
measure, equivalent to all $\mathbb{P}\in\mathcal{P}.$ Furthermore, for any $X\in\mathcal{X}$, its price is given by $e^{-rT}E_{\mathbb{Q}}[X]$. 

\begin{remark}\label{breeden}
Under the assumption that all call options $(S_T-K)^{+}$, $K \geq 0$ are traded and under the condition that the function mapping every $K\geq0$ to the price of $(S_T-K)^{+}$ is twice
differentiable with respect to $K$, \cite{breeden1978} show that any $X \in \mathcal{X}$ can almost surely be replicated using a static portfolio of calls. 
\end{remark}

Consider an investor with a finite budget and planning horizon $T>0$ who wishes to invest in the market whilst having ambiguous views on the real-world probability measure. How can she find her optimal investment strategy?  As in \cite{schied2005optimal}, she could maximize some robust expected utility \`{a} la  \cite{gilboa1989maxmin}. The basic idea is then to look for a payoff that maximizes the worst case expected utility, reflecting the idea that the investor aims to protect against the worst whilst hoping for the best. However, it seems  easier for investors to specify the desired probability distribution of the terminal wealth
rather than a utility function (\cite{brennan1981optimal}, \cite{sharpe2000distribution}, \cite{goldstein2008choosing}). 


As in \cite{dybvig1988inefficient}, \cite{sharpe2000distribution}, \cite{vrecko2013investors}, and  \cite{bernard2013explicit}, 
we thus assume in this paper that the investor specifies a desired (cumulative)
distribution function $F_{0}$ of future terminal wealth.
Once the investor understands
which distribution function $F_{0}$ is acceptable to her, the natural
question arises as to how to find, under ambiguity, the cheapest portfolio with a distribution
function at maturity that is ``at least as good'' as $F_{0}$. This is the robust
cost-efficiency problem formalized hereafter. 
In this regard, we need to recall the concept of integral stochastic
ordering, see e.g., \cite{denuit2005}. In this paper, we denote by $\mathbb{F}$ a set of measurable functions from $\mathbb{R}_+$ to $\mathbb{R}$.
\begin{definition}[Stochastic integral ordering]
\label{def:integral_ordering}Let $G$ and $F$ be two distribution functions with support on $\mathbb{R_+}$.
The distribution function $G$ dominates the distribution function $F$ in \emph{integral stochastic ordering with respect to the set $\mathbb{F}$},
in notation $F\preceq_{\mathbb{F}}G$, if 
\begin{equation*}
\forall f\in\mathbb{F},\quad\int_{\mathbb{R_+}}f(x)dF\leq\int_{\mathbb{R_+}}f(x)dG,
\end{equation*}
such that expectations are finite.

\end{definition}

Let $\mathbb{F}_{FSD}$ denote the set of all non-decreasing functions from $\mathbb{R}_+$ to $\mathbb{R}$.
The corresponding stochastic integral ordering is called \emph{first
order stochastic dominance} (FSD). Furthermore, let $\mathbb{F}_{SSD}$
denote the set of all non-decreasing and concave functions from $\mathbb{R}_+$ to $\mathbb{R}$. The corresponding
stochastic integral ordering is called  \emph{second order
stochastic dominance} (SSD). It is well-known that FSD reflects the common agreement of all investors with law-invariant increasing preferences \cite[Theorem 1]{bernard2015rationalizing}, whereas SSD order reflects  the common agreement of those who have law-invariant increasing and diversification loving preferences (risk averse investors) \cite[Corollary 2.6]{BernardSturm2023}.        
\begin{problem}[Robust cost-efficiency problem]
\label{def:The-robust-cost-efficiency}The $\mathbb{F}-$\emph{robust
cost-efficiency problem for a distribution function $F_{0}$} is defined
as 
\begin{equation}
\underset{X\in\mathcal{B}_{F_{0}}^{\mathbb{F}}}{\inf}e^{-rT}E_{\mathbb{Q}}[X],\label{eq:rob_cost_eff_sproblem_class}
\end{equation}
in which $\mathcal{B}_{F_{0}}^{\mathbb{F}}$ denotes a class of admissible
payoffs defined as 
\[
\mathcal{B}_{F_{0}}^{\mathbb{F}}=\left\{ X\in\mathcal{X}\,:\,\forall\mathbb{P}\in\mathcal{P}:\,F_{0}\preceq_{\mathbb{F}}F_{X}^{\mathbb{P}}\right\} .
\]
A solution to \eqref{eq:rob_cost_eff_sproblem_class} is called a
\emph{$\mathbb{F}-$robust cost-efficient payoff}.
\end{problem}

As discussed above, the target distribution function of the investor
is $F_{0}$. That is, we are interested in all payoffs that have
a distribution function at maturity that is at least as good as $F_{0}$ under all plausible
scenarios $\mathbb{P\in\mathcal{P}}$. For example, when $\mathbb{F}= \mathbb{F}_{FSD}$, we care about payoffs having distribution functions $F_{X}^\mathbb{P},$ $\mathbb{P} \in \mathcal{P}$ that dominate  $F_0$ in FSD. 
In order not to ``throw away
investors' money'', see \cite{dybvig1988inefficient}, we then aim to determine 
the cheapest payoff among
the payoffs in the admissible set  $\mathcal{B}_{F_{0}}^{\mathbb{F}}$.  
In Theorem~\ref{thm:robust-cost-payoff} we provide solutions to the $\mathbb{F}-$robust cost-efficiency problem \eqref{eq:rob_cost_eff_sproblem_class}  
under regularity conditions on the set $\mathbb{F}$ and $F_0$. 


\cite{dybvigJoB,dybvigRFS} introduced the standard cost-efficiency
problem without ambiguity on the set of physical measures; that is, when $\mathcal{P}=\{\mathbb{P}\}$. Specifically, for some fixed $\mathbb{P}\in\mathcal{P}$, the problem he considered reads as 

	\begin{equation}
		\underset{X\in\mathcal{A}_{F_{0}}^{\mathbb{P}}}{\inf}e^{-rT}E_{\mathbb{Q}}[X],\label{eq:cost_eff_sproblem_class}
	\end{equation}
where
\[
\mathcal{A}_{F_{0}}^{\mathbb{P}}=\left\{ X\in\mathcal{X}\,:\,F_{0}=F_{X}^{\mathbb{P}}\right\} .
\]
We refer to this problem as the \emph{standard cost-efficiency problem}. Furthermore, we say that a payoff $X$ that is distributed with $F_X^\mathbb{P}$ is \emph{$\mathbb{P}-$cost-efficient} if $X$
solves the standard cost-efficiency problem (\ref{eq:cost_eff_sproblem_class})
under $\mathbb{P}$ w.r.t. $F_0=F_X^\mathbb{P}$. By \cite{bernard2013explicit}, a payoff $X$ is $\mathbb{P}-$cost-efficient if and only if $X$ is non-increasing in the state price $\mathbb{P}-$a.s., see also \citet[Proposition 2.5]{schied2004neyman}.


The standard cost-efficiency problem (\ref{eq:cost_eff_sproblem_class}) 
has been solved in \cite{dybvigJoB,dybvigRFS} and in \cite{bernard2013explicit},
see also Lemma~\ref{lem:N_cost_eff_P} 
in the Appendix. 
In Corollary~\ref{cor:singleton} we show that in the case in which  $\mathcal{P}$
is a singleton, the solution to the $\mathbb{F}-$robust cost-efficiency problem is unique and coincides with the solution to the standard cost-efficiency problem. 

\begin{remark}
	\label{rem:allPriors}For a fixed $X\in\mathcal{X}$, we have that in general $F_{X}^{\mathbb{P}}$ cannot be equal to $F_{0}$ for
	all $\mathbb{P}\in\mathcal{P}$. Therefore, we replace the condition
	$F_{0}=F_{X}^{\mathbb{P}}$ in the standard cost-efficiency problem
	with $F_{0}\preceq_{\mathbb{F}}F_{X}^{\mathbb{P}}$ in the robust setting.
\end{remark}


The next example anticipates Sections~\ref{subsec:Example:-Uncertainty-about-drift}
and~\ref{subsec:Examples} and is designed to help distinguish
between the standard and the robust cost-efficiency problem.
\begin{example}(Cost-efficient payoffs).
\label{exa:standard_vs_robust}Assume the real-world distribution
function of $S_{T}$ is log-normal. There are three investors: one
investor assumes that the drift of $S_{T}$ under the physical measure,
denoted by $\mathbb{P}^{\mu_{1}}$, is equal to $\mu_{1}>r$. Another
investor assumes that the drift is given by $\mu_{2}>\mu_{1}$ under the
physical measure, denoted by $\mathbb{P}^{\mu_{2}}$. A third investor
has ambiguity and assumes that the drift lies in the interval $[\mu_{1},\mu_{2}]$,
and thus considers a set $\mathcal{P}=\{\mathbb{P}^{\mu},\,\mu\in[\mu_{1},\mu_{2}]\}$
as the set of all plausible probability measures on $(\Omega,\mathcal{F})$.
The cheapest payoffs to obtain a fixed target distribution function
$F_{0}$ are well-known for investor one and two and are given by
\[
X_{1}^{\ast}:=F_{0}^{-1}\left(F_{S_{T}}^{\P^{\mu_{1}}}(S_{T})\right)\quad\text{and}\quad X_{2}^{\ast}:=F_{0}^{-1}\left(F_{S_{T}}^{\P^{\mu_{2}}}(S_{T})\right),
\]
respectively; see Proposition 3 in \cite{bernard2013explicit}. Within
the set $\mathcal{P}$, $\mathbb{P}^{\mu_{1}}$ corresponds to a pessimistic view of the stock price behavior, and
we will see in Section~\ref{subsec:Example:-Uncertainty-about-drift}
that $X_{1}^{\ast}$ also solves the robust cost-efficiency problem
for $F_{0}$ if $\mathbb{F}=\mathbb{F}_{FSD}$. In the case in which $\mathbb{F}=\mathbb{F}_{SSD}$,
$X_{1}^{\ast}$ also solves the robust cost-efficiency problem if additionally
$F_{0}^{-1}\circ F_{S_{T}}^{\P^{\mu_{1}}}$ is concave; see Section
\ref{subsec:Examples}.  This example
illustrates that the solution to the standard cost-efficiency problem
for arbitrarily $\mathbb{P}\in\mathcal{P}$ and the solution
to the robust cost-efficiency problem generally do not coincide.
\end{example}
%

\begin{remark}[{Uncertainty on target distribution $F_0$}]
\label{rem:ambi_F0}As in \cite{ruschendorf2016method}, we could
also consider uncertainty on the target distribution function $F_{0}.$
Specifically, in \cite{ruschendorf2016method} it is assumed that the
investor specifies finitely many acceptable distribution functions
$F^{1}_0,...,F^{N}_0$. As all $N$ distribution functions are acceptable
to the investor, she could solve the robust cost-efficiency problem
$N$ times and buy the cheapest solution among the $N$ solutions. 
It is also possible to consider a continuum
of acceptable distribution functions: generalizing \cite{ruschendorf2016method}
slightly, we could consider the set 
\[
\mathcal{G}=\left\{ F\ | F\ \text {is a cdf},\quad F_{0}\preceq_{\mathbb{F}}F\right\} ,
\]
as the set of distribution functions that are acceptable to the investor or
client. We could then consider the cost-efficiency problem with uncertainty
on the physical measure \emph{and} the target distribution by
\begin{equation}
\underset{X\in\mathcal{B}_{\mathcal{G}}^{\mathbb{F}}}{\inf}e^{-rT}E_{\mathbb{Q}}[X],\quad\mathcal{B}_{\mathcal{G}}^{\mathbb{F}}=\left\{ X\in\mathcal{X}\,| \, \ \forall\mathbb{P}\in\mathcal{P}\,\  \exists F\in\mathcal{G},\,F\preceq_{\mathbb{F}}F_{X}^{\mathbb{P}}\right\} .\label{eq:robsut_uncer_F0}
\end{equation}
However, if $\preceq_{\mathbb{F}}$ is transitive, it holds that $\mathcal{B}_{\mathcal{G}}^{\mathbb{F}}=\mathcal{B}_{F_{0}}^{\mathbb{F}}$,
i.e., 
problem (\ref{eq:rob_cost_eff_sproblem_class})
without uncertainty on the target distribution and problem (\ref{eq:robsut_uncer_F0})
are equivalent.
\end{remark}

\subsection{\label{subsec:Technical-assumptions} Assumptions}

In order to solve the robust cost 
efficiency problem (\ref{eq:rob_cost_eff_sproblem_class}), we need some regularity conditions on the set $\mathbb{F}$ and on the target distribution $F_0$. In this regard, we define some concepts.

Recall first the concept of a \emph{least favorable measure} introduced
by \cite{schied2005optimal} for the case $\mathbb{F}=\mathbb{F}_{FSD}.$ For $\mathbb{P}\in\mathcal{P}$, we define the corresponding \emph{likelihood ratio}\footnote{
		The random variable $\frac{e^{-rT}}{\ell^{\mathbb{P}}}$ is also called
		\emph{state price} because the price of a payoff $X\in\mathcal{X}$
		can be expressed by
		\[
		e^{-rT}E_{\mathbb{Q}}[X]=E_{\mathbb{P}}\left[\frac{e^{-rT}}{\ell^{\mathbb{P}}}X\right],\quad\mathbb{P}\in\mathcal{P}.
		\]
	} by $\ell^{\mathbb{P}}=\frac{d\mathbb{P}}{d\mathbb{Q}}$.

\begin{definition}[Least favorable measure with respect to $\mathbb{F}$]
\label{Least-favorable-measure}A measure $\mathbb{P}^{\ast}\in\mathcal{P}$ with corresponding
likelihood ratio $\ell^{\ast}:=\frac{d\mathbb{P}^\ast}{d\mathbb{Q}}$ is called a \emph{least favorable
measure with respect to $\mathbb{F}$ }if $F_{\ell^{\ast}}^{\mathbb{P}^{\ast}}\preceq_{\mathbb{F}}F_{\ell^{\ast}}^{\mathbb{P}}$
for all $\mathbb{P}\in\mathcal{P}$. 
\end{definition}

Definition~\ref{Least-favorable-measure} generalizes Definition 2.1
of \cite{schied2005optimal}, who assumed the existence of a least favorable measure  w.r.t.\, $\mathbb{F}_{FSD}$ to determine payoffs that solve the robust expected utility problem of Gilboa-Schmeidler. We also need the following definition.

\begin{definition}[Composition-consistency of $\mathbb{F}$]
The set $\mathbb{F}$ is said to be \emph{composition-consistent}
if for $f,g\in\mathbb{F}$ also $f\circ g\in\mathbb{F}$. 
\end{definition}

Note that the sets $\mathbb{F}_{FSD}$ and $\mathbb{F}_{SSD}$ are composition-consistent.
This follows from the fact that the composition of non-decreasing
(resp. non-decreasing and concave) functions is again non-decreasing
(resp. non-decreasing and concave). 

The following proposition provides conditions that guarantee the existence of a least favorable measure and turns out to be very useful for applications.  
\begin{proposition}[Sufficient conditions for the existence of a least favorable measure]
\label{cor:for examples}Assume that $\mathbb{F}$ is composition
consistent. If $F_{S_{T}}^{\mathbb{P}^{\prime}}\preceq_{\mathbb{F}}F_{S_{T}}^{\mathbb{P}}$
for some $\mathbb{P}^{\prime}\in\mathbb{P}$ and all $\mathbb{P}\in\mathcal{P},$
and $\ell^{\mathbb{P}^{\prime}}=f(S_{T})$ for some $f\in\mathbb{F}$,
then $\mathbb{P}^{\prime}$ is a least favorable measure w.r.t. $\mathbb{F}$. If, additionally,
$S_{T}$ is continuously distributed under $\mathbb{P}^{\prime}$
and $f$ is strictly increasing, then $\ell^{\mathbb{P}^{\prime}}$
is continuously distributed under $\mathbb{P}^{\prime}$.
\end{proposition}

\begin{proof}
Let $\mathbb{P},\mathbb{P}^{\prime}\in\mathcal{P}$. Let $X$ be a
payoff and $f\in\mathbb{F}$. Note
that $F_{X}^{\mathbb{P}^{\prime}}\preceq_{\mathbb{F}}F_{X}^{\mathbb{P}}$
if and only if $E_{\mathbb{P}^{\prime}}[g(X)]\leq E_{\mathbb{P}}[g(X)]$
for all $g\in\mathbb{F}$ such that expectations are finite. Because
$\mathbb{F}$ is composition-consistent, it follows that
\begin{equation}
F_{X}^{\mathbb{P}^{\prime}}\preceq_{\mathbb{F}}F_{X}^{\mathbb{P}}\Rightarrow F_{f(X)}^{\mathbb{P}^{\prime}}\preceq_{\mathbb{F}}F_{f(X)}^{\mathbb{P}}.\label{eq:F_X =00003D> F_f(X)}
\end{equation}
The expression $\ensuremath{F_{\ell^{\mathbb{P}^{\prime}}}^{\mathbb{P}^{\prime}}\preceq_{\mathbb{F}}F_{\ell^{\mathbb{P}^{\prime}}}^{\mathbb{P}}}$
then follows by Equation (\ref{eq:F_X =00003D> F_f(X)}). Let $f$ be strictly increasing. By \cite{embrechts2013note},
the generalized inverse $f^{-1}$ of $f$ is continuous on the range of $f$.
Thus, it holds that 
\begin{align*}
\mathbb{P}^{\prime}\left(\ell^{\mathbb{P}^{\prime}}\leq x\right) & =\mathbb{P}^{\prime}\left(S_{T}\leq f^{-1}(x)\right)=F_{S_{T}}^{\mathbb{P}^{\prime}}\left(f^{-1}(x)\right),\quad x\in\mathbb{R},
\end{align*}
which implies that $\ell^{\mathbb{P}^{\prime}}$ is continuously distributed
under $\mathbb{P}^{\prime}$.
\end{proof}

We also need a definition that is, to the best of our knowledge, new
to the literature.
\begin{definition}[Cost-consistency of $\mathbb{F}$]
\label{def:cost_consistant}The set $\mathbb{F}$
is called \emph{cost-consistent }if for all $X,Y\in\mathcal{X}$ and
all $\mathbb{P}\in\mathcal{P}$ such that $X,Y$ are $\mathbb{P}-$cost-efficient,
$F_{X}^{\mathbb{P}}\preceq_{\mathbb{F}}F_{Y}^{\mathbb{P}}$ implies
$E_{\mathbb{Q}}[X]\leq E_{\mathbb{Q}}[Y]$ and, additionally, $F_{X}^{\mathbb{P}}\neq F_{Y}^{\mathbb{P}}$
implies $E_{\mathbb{Q}}[X]<E_{\mathbb{Q}}[Y]$. 
\end{definition}

As the set $\mathbb{F}_{SSD}$ is contained in $\mathbb{F}_{FSD}$,
the following proposition implies that $\mathbb{F}_{FSD}$ and $\mathbb{F}_{SSD}$
are cost-consistent. In Example~\ref{lem:TSD} we discuss a set that
is not cost-consistent.
\begin{proposition}
\label{lem:cost-consistent}
$\mathbb{F}_{SSD}$ is cost-consistent. Moreover, if $\mathbb{F}_{SSD}\subset\mathbb{F}$, then $\mathbb{F}$
is cost-consistent. 
\end{proposition}

\begin{proof}
The cost-consistency of $\mathbb{F}_{SSD}$ can be proven along the lines of the proof of Lemma 2 in \cite{bernard2019optimal}. Furthermore,
$F_{X}^{\mathbb{P}}\preceq_{\mathbb{F}}F_{Y}^{\mathbb{P}}$ implies
$F_{X}^{\mathbb{P}}\preceq_{\mathbb{F}_{SSD}}F_{Y}^{\mathbb{P}}$,
which finishes the proof. 
\end{proof}

We now list a series of conditions that we often use to derive our main results: 

\begin{assumption}
	\label{assu:F0^-1}$F_{0}^{-1}$ is square integrable, i.e., $\int_{0}^{1}(F_{0}^{-1}(u))^{2}du<\infty$, and $F_{0}(x)=0$ for $x<0$.
\end{assumption}

\begin{assumption}
	\label{assu:setF}The set $\mathbb{F}$ is composition-consistent
	and cost-consistent.
\end{assumption}

\begin{assumption}
	\label{assu:least_fav} The set $\mathcal{P}$ contains a least favorable measure with respect to $\mathbb{F}.$ We denote this least favorable measure by $\mathbb{P}^{\ast}.$  
\end{assumption}

\begin{assumption}
	\label{assu:l*}Denote by $\ell^{\ast}$ the likelihood ratio of the
	least favorable measure $\mathbb{P}^{\ast}.$ We assume that 
	$x\mapsto F_{\ell^{\ast}}^{\mathbb{P}^{\ast}}(x)$ is continuous and
	that $\frac{1}{\ell^{\ast}}$ has finite variance under $\mathbb{P}^{\ast}$.
\end{assumption}

Condition~\ref{assu:F0^-1} is technical and ensures that the robust cost-efficiency
problem is well-posed, i.e., $\mathcal{B}_{F_{0}}^{\mathbb{F}}$ is
not empty, see Theorem~\ref{thm:robust-cost-payoff}. 

Condition~\ref{assu:least_fav} can also be found in \cite{schied2005optimal} for the case $\mathbb{F}=\mathbb{F}_{FSD}.$ Note that when $\mathbb{F}$ becomes larger, the condition~\ref{assu:least_fav} becomes stronger. Specifically, requiring a least favorable measure $\mathbb{P}^{\ast}\in\mathcal{P}$
with respect to $\mathbb{F}=\mathbb{F}_{FSD}$ is more stringent than in the case  $\mathbb{F}=\mathbb{F}_{SSD}$. 
In particular, Proposition
\ref{cor:for examples} provides sufficient conditions for the existence of a least favorable measure ${P}^{\ast} \in \mathcal{P}$ with respect to $\mathbb{F}$. 

The condition in~\ref{assu:l*} that $x\mapsto F_{\ell^{\ast}}^{\mathbb{P}^{\ast}}(x)$ is continuous 
distribution function under $\mathbb{P}^{\ast}$ 
is also made in a
setting without ambiguity in e.g., \cite{JZ08},  \cite{he2011portfolio,he2011portfoliob},  \cite{bernard2013explicit} and \cite{xu2016note} among many others. It is a strong assumption in the sense that we essentially exclude discrete settings.


\section{\label{sec:Robust-cost-efficiency}Solution of the robust cost-efficiency
problem}

In the next theorem, we make the assumption that $F_{0}^{-1}\circ F_{\ell^{\ast}}^{\mathbb{P}^{\ast}}\in\mathbb{F}$.
Note that this assumption is always true if $\mathbb{F}=\mathbb{F}_{FSD}$.
The assumption is also true if $\mathbb{F}=\mathbb{F}_{SSD}$, provided
that $F_{0}^{-1}\circ F_{\ell^{\ast}}^{\mathbb{P}^{\ast}}$ is concave.
\begin{thm}[$\mathbb{F}-$robust cost-efficient payoff]
\label{thm:robust-cost-payoff} Assume that the conditions 
\ref{assu:F0^-1},~\ref{assu:setF},~\ref{assu:least_fav} and~\ref{assu:l*} hold and that $F_{0}^{-1}\circ F_{\ell^{\ast}}^{\mathbb{P}^{\ast}}\in\mathbb{F}.$
Then, the $\mathbb{F}-$robust cost-efficiency problem for $F_{0}$
has a $\mathbb{P}^{\ast}-$a.s. unique solution given by 
\[
F_{0}^{-1}\left(F_{\ell^{\ast}}^{\mathbb{P}^{\ast}}\left(\ell^{\ast}\right)\right).
\]
\end{thm}

\begin{proof}
Recall that $\ell^{\ast}$ denotes the likelihood ratio that corresponds to
$\mathbb{P}^{\ast}$. Let 
\begin{align*}
X^{\ast}=F_{0}^{-1}\left(F_{\ell^{\ast}}^{\mathbb{P}^{\ast}}(\ell^{\ast})\right).
\end{align*}
As $F_{\ell^{\ast}}^{\mathbb{P}^{\ast}}\left(\ell^{\ast}\right)$
is uniformly distributed under $\mathbb{P}^{\ast}$ (Condition ~\ref{assu:l*}), it follows by Lemma \ref{lem:N_cost_eff_P} that $F_{X^{\ast}}^{\mathbb{P}^{\ast}}=F_{0}$;
and, by condition~\ref{assu:F0^-1} it holds that
\[
E_{\mathbb{P}^{\ast}}\left[\left(F_{0}^{-1}\left(F_{\ell^{\ast}}^{\mathbb{P}^{\ast}}\left(\ell^{\ast}\right)\right)\right)^{2}\right]=\int_{0}^{1}(F_{0}^{-1}(u))^{2}du<\infty.
\]
Condition~\ref{assu:l*} thus implies that $E_{\mathbb{Q}}[X^{\ast}]<\infty$
because
\begin{align*}
E_{\mathbb{Q}}[|X^{\ast}|] & =E_{\mathbb{P}^{\ast}}\left[\frac{1}{\ell^{\ast}}F_{0}^{-1}\left(F_{\ell^{\ast}}^{\mathbb{P}^{\ast}}\left(\ell^{\ast}\right)\right)\right]\\
 & \leq\sqrt{E_{\mathbb{P}^{\ast}}\left[\frac{1}{(\ell^{\ast})^{2}}\right]E_{\mathbb{P}^{\ast}}\left[\left(F_{0}^{-1}\left(F_{\ell^{\ast}}^{\mathbb{P}^{\ast}}\left(\ell^{\ast}\right)\right)\right)^{2}\right]}\\
 & <\infty.
\end{align*}
Therefore, $X^{\ast}\in\mathcal{X}$. By conditions~\ref{assu:setF} and~\ref{assu:least_fav}
and as $F_{0}^{-1}\circ F_{\ell^{\ast}}^{\mathbb{P}^{\ast}}\in\mathbb{F}$,
it follows from Equation (\ref{eq:F_X =00003D> F_f(X)}) that $F_{0}=F_{X^{\ast}}^{\mathbb{P^{\ast}}}\preceq_{\mathbb{F}}F_{X^{\ast}}^{\mathbb{P}}$
for all $\mathbb{P}\in\mathcal{P}$; hence, $X^{\ast}\in\mathcal{B}_{F_{0}}^{\mathbb{F}}$.
Let $Y\in\mathcal{B}_{F_{0}}^{\mathbb{F}}$ and define
\[
Y^{\ast}=\left[F_{Y}^{\mathbb{P}^{\ast}}\right]^{-1}\left(F_{\ell^{\ast}}^{\mathbb{P}^{\ast}}(\ell^{\ast})\right).
\]
Then $Y^{\ast}$ is $\mathbb{P}^{\ast}-$cost-efficient for $F_{Y}^{\mathbb{P}^{\ast}}$
and we have $F_{X^{\ast}}^{\mathbb{P}^{\ast}}=F_{0}\preceq_{\mathbb{F}}F_{Y}^{\mathbb{P}^{\ast}}=F_{Y^{\ast}}^{\mathbb{P}^{\ast}}$.
By condition~\ref{assu:setF}, $\mathbb{F}$ is cost-consistent,
which implies that $E_{\mathbb{Q}}[X^{\ast}]\leq E_{\mathbb{Q}}[Y^{\ast}]\leq E_{\mathbb{Q}}[Y]$.
Hence, every admissible payoff is more expensive than $X^{\ast}$.
We now show uniqueness. Let $\hat{X}$ be another solution to the
robust cost-efficiency problem. It holds that $F_{0}\preceq_{\mathbb{F}}F_{\hat{X}}^{\mathbb{P}^{\ast}}$.
If $F_{\hat{X}}^{\mathbb{P}^{\ast}}=F_{0}$ and $E_{\mathbb{Q}}[X^{\ast}]=E_{\mathbb{Q}}[\hat{X}]$,
then $X^{\ast}=\hat{X}$, $\mathbb{P}^{\ast}-$a.s. by Lemma~\ref{lem:N_cost_eff_P}
because the solution $X^{\ast}$ corresponds to the solution of the
standard $\mathbb{P}^{\ast}-$cost-efficiency problem for $F_{0}$,
which has a unique solution. If $F_{0}\neq F_{\hat{X}}^{\mathbb{P}^{\ast}}$,
then $E_{\mathbb{Q}}[X^{\ast}]<E_{\mathbb{Q}}[\hat{X}]$ because $\mathbb{F}$
is cost-consistent. Hence, $X^{\ast}$ is the unique solution to the
robust cost-efficiency problem. 
\end{proof}
\begin{remark}
\label{rem:Instead-of-requiring}Instead of requiring that $F_{0}^{-1}$
is square integrable, the proof of Theorem~\ref{thm:robust-cost-payoff}
shows that it is sufficient to assume that $F_{0}^{-1}\left(F_{\ell^{\ast}}^{\mathbb{P}^{\ast}}\left(\ell^{\ast}\right)\right)$
has finite price.
\end{remark}

\begin{remark}
\label{rem:AB}Does ambiguity increase costs? Let the assumptions
of Theorem~\ref{thm:robust-cost-payoff} be in force. Let us compare
two investors. Investor A has ambiguity and considers the set $\mathcal{P}$
as the set of possible real-world measures. Investor B has, e.g., based
on a deep market analysis or insider knowledge, no ambiguity and knows
that $\mathbb{P}\in\mathcal{P}$ is the true real-world measure. Both
investors consider $F_{0}$ as the target distribution function. Investor
A buys $X^{\ast}=F_{0}^{-1}\left(F_{\ell^{\ast}}^{\mathbb{P}^{\ast}}\left(\ell^{\ast}\right)\right)$
according to Theorem~\ref{thm:robust-cost-payoff}, whereas investor B buys
$X=F_{0}^{-1}(F_{\ell^{\mathbb{P}}}^{\mathbb{P}}(\ell^{\mathbb{P}}))$
(see Lemma~\ref{lem:N_cost_eff_P}). As $X^{\ast}\in\mathcal{B}_{F_{0}}^{\mathbb{F}}$,
it holds that $F_{X}^{\mathbb{P}}=F_{0}\leq_{\mathbb{F}}F_{X^{\ast}}^{\mathbb{P}}$.
As the set $\mathbb{F}$ is cost-consistent, it follows that $E_{\mathbb{Q}}[X]\leq E_{\mathbb{Q}}[X^{\ast}]$.
If we additionally have $F_{X}^{\mathbb{P}}\neq F_{X^{\ast}}^{\mathbb{P}}$,
then it follows that $E_{\mathbb{Q}}[X]<E_{\mathbb{Q}}[X^{\ast}]$. In
the robust setting we end up with a payoff $X^\ast$ whose distribution $F_{X^\ast}^\mathbb{P}$, $\mathbb{P} \in \mathcal P$ dominates $F_{0}$ in stochastic
ordering for all $\mathbb{P}$. That is, under ambiguity, the preferred payoff has a (strictly)
higher price and the optimal robust choice $X^\ast$ typically will not match the
choice $X$ without uncertainty. 
\end{remark}

\begin{remark}
\label{remark}
The condition in Theorem~\ref{thm:robust-cost-payoff} that the function
$F_{0}^{-1}\circ F_{\ell^{\ast}}^{\mathbb{P}^{\ast}}$ must be concave
in the case in which $\mathbb{F}=\mathbb{F}_{SSD}$ means that the target distribution
function $F_{0}$ is required to be lighter-tailed than the distribution
function $F_{\ell^{\ast}}^{\mathbb{P}^{\ast}}$. Specifically, $F_{\ell^{\ast}}^{\mathbb{P}^{\ast}}$
must dominate $F_{0}$ in the sense of transform convex order (\cite{shaked2007stochastic}).
\end{remark}

The next corollary shows that the standard and the robust cost-efficiency
problem coincide in a setting without uncertainty. Note that we do
not require $F_{0}^{-1}\circ F_{\ell^{\ast}}^{\mathbb{P}^{\ast}}\in\mathbb{F}$
as in Theorem~\ref{thm:robust-cost-payoff}.
\begin{corollary}
\label{cor:singleton}Assume that the conditions ~\ref{assu:F0^-1} and~\ref{assu:l*} 
hold and that $\mathcal{P}=\{\mathbb{P}\}$ is a singleton. 
The solutions to the $\mathbb{F}-$robust cost-efficiency problem
for $F_{0}$  and the standard cost-efficiency problem for $F_{0}$
are unique and identical.
\end{corollary}

\begin{proof}
Note that when $\mathcal{P}=\{\mathbb{P}\}$, then $\mathbb{P}$ is a least favorable measure. By Lemma~\ref{lem:N_cost_eff_P}, $X^{\ast}=F_{0}^{-1}(F_{\ell^{\mathbb{P}}}^{\mathbb{P}}(\ell^{\mathbb{P}})$)
is the unique solution to the standard cost-efficiency problem. As
in the proof of Theorem~\ref{thm:robust-cost-payoff}, one can show
that $X^{\ast}\in\mathcal{X}$. Then $X^{\ast}\in\mathcal{B}_{F_{0}}^{\mathbb{F}}$
follows immediately because $F_{X^{\ast}}^{\mathbb{P}}=F_{0}$. As
in the proof of Theorem~\ref{thm:robust-cost-payoff}, one can show
that $X^{\ast}$ is the only admissible payoff solving  the $\mathbb{F}-$robust cost-efficiency problem.
\end{proof}

The sets $\mathbb{F}_{FSD}$ and $\mathbb{F}_{SSD}$ are cost- and
composition-consistent. We provide examples of sets $\mathbb{F}$
of functions that are not cost- or composition-consistent so that
Theorem~\ref{thm:robust-cost-payoff} cannot be applied to find robust
optimal payoffs.
\begin{example}
\label{lem:TSD} Third order stochastic dominance is the stochastic
integral ordering that arises from the set $\mathbb{F}_{TSD}$, containing
all functions $f:\mathbb{R}_+\to\mathbb{R}$ such that $f^{\prime}>0$, $f^{\prime\prime}\leq0$
and $f^{\prime\prime\prime}\geq0$. 
The set $\mathbb{F}_{TSD}$ is composition-consistent but is in general
not cost-consistent: see Appendix \ref{sec:Proof-of-TSD2}.
\end{example}

\begin{example}
\cite{muller2017between} introduced the $(1+\gamma)-$stochastic
dominance order for $\gamma\in(0,1)$, which lies between FSD and
SSD ordering. The set induced by $(1+\gamma)-$stochastic dominance
order is in general not composition-consistent, but it is cost-consistent
in light of Proposition~\ref{lem:cost-consistent}. 
\end{example}

\begin{example}
\cite{rothschild1970increasing} introduced concave stochastic order,
which is defined via the set of all concave (but not necessarily non-decreasing)
functions. Concave stochastic order coincides with SSD if we compare
two payoffs with the same mean \cite[Remark 2.63]{follmer2011stochastic}.
The set of all concave functions is cost-consistent but not composition-consistent.
\end{example}

In the following section we illustrate Theorem~\ref{thm:robust-cost-payoff}
in a log-normal market setting with uncertainty on the drift \emph{and} volatility,
whereas in Section~\ref{subsec:Example:-General-market} we deal with a more general market setting.

\subsection{\label{subsec:Example:-Uncertainty_log_normal} Robust cost-efficient
payoffs in lognormal markets}

We assume that 
under the pricing measure $\Q$, $S_{T}$ has a log-normal
distribution function with parameters $\log(S_{0})+(r-\frac{{s}^{2}}{2})T$
and ${s}^{2}T$ with stock price $S_{0}>0$ today, interest rate
$r\in\mathbb{R}$, time horizon $T>0$ and volatility $s>0$. Under
$\mathbb{Q}$, $S_{T}$ is log-normally distributed with density $f^{r,s}$,
where for $m\in\mathbb{R}$ and $\varsigma>0$ we define
\begin{equation}
f^{m,\varsigma}(x)=\frac{1}{x\varsigma\sqrt{T}\sqrt{2\pi}}\exp\left(-\frac{\left(\ln\left(x\right)-\ln(S_{0})-\left(m-\frac{\varsigma^{2}}{2}\right)T\right)^{2}}{2\varsigma^{2}T}\right),\quad x>0.\label{eq:f^m,sig}
\end{equation}

\subsubsection{\label{subsec:Example:-Uncertainty-about-drift}Drift uncertainty:
$\mathbb{F}_{FSD}-$robust cost-efficient payoff \citep{schied2005optimal}}

The real-world distribution function of $S_{T}$ is assumed to be
log-normal with parameters $\log(S_{0})+(\mu-\frac{s^{2}}{2})T$ and
$s^{2}T$, but there is uncertainty about the precise level of the
drift parameter $\mu$. In particular, the agent only expects the
true drift parameter $\mu$ to lie in the interval $\mathcal{D}^{\mu_{1}}=\left\{ \mu\in\mathbb{R}\,:\,\mu\geq\mu_{1}\right\} $
for $\mu_{1}>r$, and thus she considers $\mathcal{P}=\left(\mathbb{P}^{\mu}\right)_{\mu\in\mathcal{D}^{\mu_{1}}}$
as the set of all plausible probability measures on $(\Omega,\mathcal{F})$.
Under $\P^{\mu}$, $S_{T}$ is log-normal with density $f^{\mu,s}$.
It follows that $F_{S_{T}}^{\P^{\ast}}\preceq_{\mathbb{F}_{FSD}}F_{S_{T}}^{\P^{\mu}}$
for all $\mu\geq\mu_{1}$, where $\mathbb{P}^{\ast}:=\P^{\mu_{1}}$.
Let $h^{\mu,\varsigma}(x)=\frac{f^{\mu,\varsigma}(x)}{f^{r,s}(x)}$,
$x>0$. A straightforward computation shows that 
\begin{equation}
h^{\mu,s}(x)=\left(\frac{x}{S_{0}}\right)^{\frac{\left(\mu-r\right)}{s^{2}}}\exp\left(\frac{r^{2}-\mu^{2}+s^{2}(\mu-r)}{2s^{2}}T\right),\quad x>0.\label{eq:h^mu,s}
\end{equation}
As $\ell^{\P^{\mu_{1}}}= h^{\mu_{1},s}(S_{T})$ and $\mu_{1}>r,$
$\ell^{\P^{\mu_{1}}}$ is a strictly increasing function
of $S_{T}$. Furthermore, $\frac{1}{\ell^{\P^{\mu_{1}}}}$
has finite variance. By Proposition~\ref{cor:for examples}, $\mathbb{P}^{\ast}$
is a least favorable measure with corresponding likelihood ratio $\ell^{\ast}:=\ell^{\P^{\mu_{1}}}$,
i.e., conditions~\ref{assu:least_fav} and~\ref{assu:l*} are satisfied.
Theorem~\ref{thm:robust-cost-payoff} shows that the $\mathbb{F}_{FSD}-$robust cost-efficient
payoff for a distribution function $F_{0}$ satisfying condition
\ref{assu:F0^-1} is given by 
\[
X^{\ast}=F_{0}^{-1}\left(F_{\ell^{\ast}}^{\mathbb{P}^{\ast}}\left(\ell^{\ast}\right)\right)=F_{0}^{-1}\left(F_{S_{T}}^{\P^{\mu_{1}}}(S_{T})\right).
\]

The second equality follows from the increasingness of $\ell^{\P^{\mu_{1}}}$
in $S_{T}$. The agent thus chooses the optimal payoff as if she believes
that the worst-case plausible value for the drift parameter $\mu$,
i.e., $\mu_{1}$, will materialize. This finding is consistent with
the results obtained by \citet[Section 3.1]{schied2005optimal} on the
impact of drift uncertainty on optimal payoff choice in a Black-Scholes
setting.
\begin{example}
We consider next the exponential distribution for the distribution
function $F_{0}$, i.e., $F_{0}(x)=1-e^{-x}$, $x\geq 0$, which satisfies condition
\ref{assu:F0^-1}. Panel A of Figure~\ref{fig:Price-of-the} displays
the price of the robust cost-efficient payoff for varying levels of
the parameter $\mu_{1}$, which describes the ambiguity that the agent
faces (consistently with Remark~\ref{rem:AB}). The higher $\mu_{1}$, the smaller the set $\mathcal{D}^{\mu_{1}}$, i.e., the lower the degree of ambiguity,
and the cheaper $X^{\ast}$. In panel B of Figure~\ref{fig:Price-of-the}
we display, for several values of $\mu_{1}$, the robust cost-efficient
payoff normalized for its initial price as a function of realizations $s$ of $S_{T}$;
i.e., we display the curve 
\[
s\mapsto\frac{F_{0}^{-1}\left(F_{S_{T}}^{\P^{\mu_{1}}}(s)\right)}{\pi_{\mu_{1}}},
\]
where $\pi_{\mu_{1}}=e^{-rT}E_{\mathbb{Q}}\left[F_{0}^{-1}\left(F_{S_{T}}^{\P^{\mu_{1}}}(S_{T})\right)\right]$.
We observe that the curve is flatter when $\mu_{1}$ is smaller, i.e.,
more ambiguity gives rise to payoffs that reflect a higher degree of conservatism. 
\end{example}

\begin{figure}[!h]
\begin{centering}
\begin{tabular}{cc}
\includegraphics[width=8.125cm,height=7.5cm]{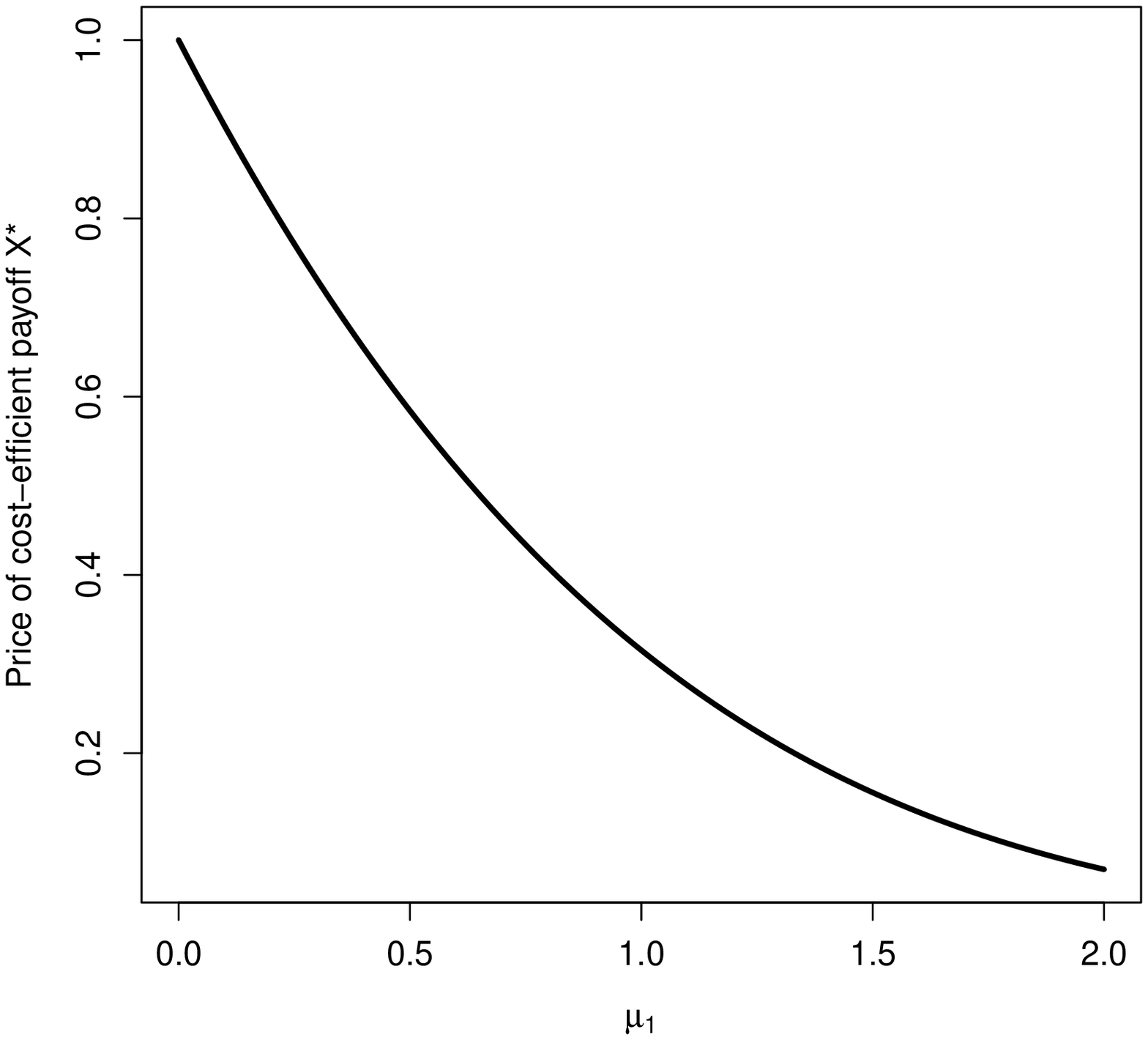} & \includegraphics[width=8.125cm,height=7.5cm]{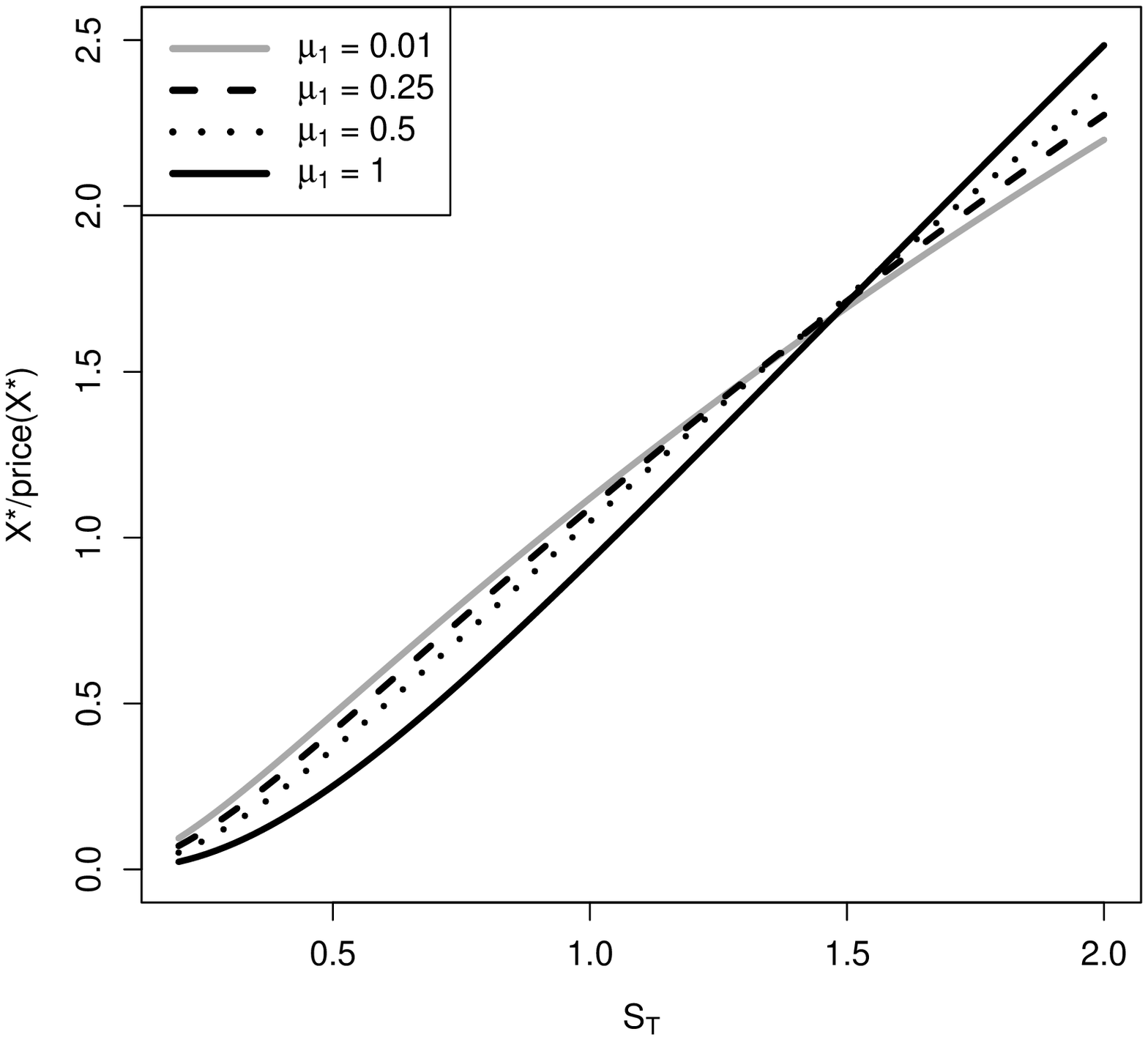}\tabularnewline
Panel A & Panel B\tabularnewline
\end{tabular}
\par\end{centering}
\caption{\label{fig:Price-of-the} We use the parameters $S_{0}=1$, $r=0$,
$T=1$ and $s=0.9$. The reference distribution function is the exponential
distribution function $F_{0}(x)=1-e^{-x}$. Panel A: Price of the
cost-efficient payoff $X^{\ast}$ depending on the value of the ambiguity
parameter $\mu_{1}$. Panel B: Cost-efficient payoff per unit of investment
for various values of $\mu_{1}$.}
\end{figure}


\subsubsection{\label{subsec:Examples}Drift and volatility uncertainty: $\mathbb{F}_{SSD}-$robust
cost-efficient payoff}

The real-world distribution function of $S_{T}$ is assumed to be
log-normal with parameters $\log(S_{0})+(\mu-\frac{{\sigma}^{2}}{2})T$
and $\sigma^{2}T$, but now the agent faces uncertainty about the
precise level of the parameters $\mu$ and $\sigma$. In particular,
the agent only expects the true parameters to lie within the cube
\[
\mathcal{D}^{\mu_{1},\mu_{2},\sigma_{1},s}=\left\{ (\mu,\sigma)\subset\mathbb{R}^{2}\,:\,\mu_{1}\leq\mu\leq\mu_{2},\,\sigma_{1}\leq\sigma\leq s\right\} 
\]
for $r<\mu_{1}<\mu_{2}$ and $0<\sigma_{1}\leq s$, and thus she considers
$\mathcal{P}=\left(\mathbb{P}^{\mu,\sigma}\right)_{(\mu,\sigma)\in\mathcal{D}^{\mu_{1},\mu_{2},\sigma_{1},s}}$
as the set of all plausible probability measures on $(\Omega,\mathcal{F})$.
Under $\P^{\mu,\sigma}$, $S_{T}$ is log-normal with density $f^{\mu,\sigma}$,
defined in Equation (\ref{eq:f^m,sig}). In this regard, note that
while $r<\mu_{1}$ is a natural assumption, there is some empirical
evidence for the hypothesis that $\sigma\leq s$; see Table 1 in \cite{christensen1998relation}
and Table 1 in \cite{christensen2002new}. 
\begin{remark}
In contrast to the dynamic Black-Scholes model, in which the stock
price $S_{T}$ is also log-normally distributed, we work in a static
market setting. In a dynamic Black-Scholes framework where continuous
trading is allowed at zero transaction cost, the absence of arbitrage
opportunities implies that the volatility of the stock does not change
when moving from the real-world measure to the risk-neutral measure,
i.e., there does not exist uncertainty about the volatility in a dynamic
Black-Scholes model. Here, however, we do not assume dynamic trading.
Hence, even when call option prices reflect a risk neutral distribution
function for $S_{T}$ that is log-normally distributed, the agent
may have a view on the real-world distribution that is different from
a log-normal and, in particular, may be unsure about the exact values
for drift and volatility. 
\end{remark}

In the next proposition, we assume that 
\begin{equation}
\ensuremath{\frac{\mu_{1}-r}{s^{2}}\in(0,1]}.\label{eq:mu_1-r}
\end{equation}
For example, $s\in[0.2,\infty)$ and $(\mu_{1}-r)\in(0,0.04]$ or
$s\in[0.35,\infty)$ and $(\mu_{1}-r)\in(0,0.1]$ imply Equation (\ref{eq:mu_1-r}),
that is: there are economically reasonable environments such that
Equation (\ref{eq:mu_1-r}) holds.
\begin{proposition}[$\mathbb{F}_{SSD}-$robust cost-efficient payoff]
\label{cor:log_normal}If $\frac{\mu_{1}-r}{s^{2}}\in(0,1]$, then
it holds that $F_{S_{T}}^{\mathbb{P}^{\mu_{1},s}}\preceq_{\mathbb{F}_{SSD}}F_{S_{T}}^{\mathbb{P}^{\mu,\sigma}}$,
$(\mu,\sigma)\in\mathcal{D}^{\mu_{1},\mu_{2},\sigma_{1},s}$ and conditions
\ref{assu:least_fav} and~\ref{assu:l*} are satisfied for the set
$\mathbb{F}_{SSD}$. The least favorable measure is $\mathbb{P}^{\ast}=\mathbb{P}^{\mu_{1},s}$
with corresponding likelihood ratio $\ell^{\ast}=\ell^{\P^{\mu_{1},s}}$.
The $\mathbb{F}_{SSD}-$robust cost-efficient payoff for the distribution
function $F_{0}$ satisfying condition~\ref{assu:F0^-1} such that
$F_{0}^{-1}\circ F_{\ell^{\ast}}^{\mathbb{P}^{\ast}}$ is concave
is then given by 
\begin{equation}
X^{\ast}:=F_{0}^{-1}\left(F_{S_{T}}^{\P^{\mu_{1},s}}(S_{T})\right).\label{eq:X_ast_log_normal}
\end{equation}
\end{proposition}

\begin{proof}
For a log-normal distribution function $F$ with parameters $M$ and
$V$, it holds that 
\[
\ensuremath{\int_{0}^{q}F^{-1}(p)dp=\frac{e^{M+\frac{V}{2}}}{q}\ \Phi\left(\Phi^{-1}(q)-\sqrt{V}\right),\quad q\in(0,1),}
\]
where $\Phi$ denotes the distribution function of a standard normal
random variable. It follows that 
\begin{align*}
\int_{0}^{q}\left[F_{S_{T}}^{\mathbb{P}^{\mu_{1},s}}\right]^{-1}(p)dp & =\frac{e^{\mu_{1}T}}{q}\ \Phi\left(\Phi^{-1}(q)-s\sqrt{T}\right)\\
 & \leq\frac{e^{\mu T}}{q}\ \Phi\left(\Phi^{-1}(q)-\sigma\sqrt{T}\right),\,\,\,q\in(0,1),\,\,\,\mu_{1}\leq\mu,\,\,\,\sigma\leq s.
\end{align*}
Hence, $F_{S_{T}}^{\mathbb{P}^{\mu_{1},s}}\preceq_{\mathbb{F}_{SSD}}F_{S_{T}}^{\mathbb{P}^{\mu,\sigma}}$,
$(\mu,\sigma)\in\mathcal{D}^{\mu_{1},\mu_{2},\sigma_{1},s}$. As in
Section~\ref{subsec:Example:-Uncertainty-about-drift}, let 
\[
\ell^{\P^{\mu_{1},s}}=\frac{f^{\mu_{1},s}(S_{T})}{f^{r,s}(S_{T})}=h^{\mu_{1},s}(S_{T}).
\]
Hence, the likelihood ratio $\ell^{\P^{\mu_{1},s}}$ is strictly increasing and concave in $S_{T}$ if  (\ref{eq:mu_1-r})
is satisfied. By Proposition~\ref{cor:for examples}, conditions
\ref{assu:least_fav} and~\ref{assu:l*} are satisfied for the set
$\mathbb{F}_{SSD}$ with least favorable measure $\Pa=\mathbb{P}^{\mu_{1},s}$
with likelihood ratio $\ell^{\ast}=\ell^{\P^{\mu_{1},s}}.$ As in Section~\ref{subsec:Example:-Uncertainty-about-drift},
some simple calculations and Theorem~\ref{thm:robust-cost-payoff} show
that the robust cost-efficient payoff for the distribution function
$F_{0}$ is given by (\ref{eq:X_ast_log_normal}).
\end{proof}
We provide an example for $F_0$ that make it possible to apply Proposition ~\ref{cor:log_normal} to determine ${\mathbb{F}_{SSD}}-$robust cost-efficient payoffs. 
In this regard, observe that $F_{\ell^{\ast}}^{\mathbb{P}^{\ast}}$
in Proposition~\ref{cor:log_normal} is the log-normal distribution with
parameters $\frac{1}{2}\theta^{2}$ and $\theta$ for $\theta:=\sqrt{T}\frac{\mu_{1}-r}{s}>0$.

\begin{example}
\label{exa:F0^-1 Fell* lognormal}If $F_{0}$ is the log-normal distribution
with parameters $M\in\mathbb{R}$ and $V>0$, then $F_{0}^{-1}\circ F_{\ell^{\ast}}^{\mathbb{P}^{\ast}}$
in Proposition~\ref{cor:log_normal} is concave if $V\leq\theta$ because
\[
F_{0}^{-1}\circ F_{\ell^{\ast}}^{\mathbb{P}^{\ast}}(x)=x^{\frac{V}{\theta}}\exp(-\frac{1}{2}\theta V+M),\quad x>0.
\]
\end{example}

\subsection{\label{subsec:Example:-General-market}Robust cost-efficient payoffs
in general markets using Esscher transform}

Inspired by \cite{corcuera2009implied}, let $S_{0}>0$ and $s>0$
and $Z$ be a payoff with mean zero and variance one. Under $\mathbb{Q}$,
assume that $Z$ has density $f_{Z}^{\mathbb{Q}}(x)>0$, $x\in\mathbb{R}$
and model the future stock price at date $T$ by 
\[
S_{T}=S_{0}e^{(r+\omega)T+s\sqrt{T}Z},
\]
where 
$\omega\in\mathbb{R}$ is a mean
correcting term, i.e., $\omega$ is chosen such that 
\begin{equation*}
e^{-rT}E_{\mathbb{Q}}\left[S_{T}\right]=S_{0}.
\end{equation*}
The density of $X=\log(S_{T})$ under $\mathbb{Q}$ is 
\[
f_{X}^{\mathbb{Q}}(x)=\frac{1}{s\sqrt{T}}f_{Z}^{\mathbb{Q}}\left(\frac{x-\log(S_{T})-(r+\omega)T}{s\sqrt{T}}\right),\quad x>0.
\]
The corresponding density of $S_{T}$ under $\mathbb{Q}$ is denoted
by $f_{S_{T}}^{\mathbb{Q}}$, and it holds that 
\[
f_{S_{T}}^{\mathbb{Q}}(x)=f_{X}^{\mathbb{Q}}(\log(x))\frac{1}{x},\quad x>0.
\]
Let $h^{\ast}>0$ and $\mathcal{H}\subset[h^{\ast},\infty)$ be a
set containing $h^{\ast}$ such that $E_{\mathbb{Q}}\left[(S_{T})^{h}\right]$
exists for all $h\in\text{\ensuremath{\mathcal{H}}}$. 
 Define a family
of probability measures $\mathcal{P}=\left(\mathbb{P}^{h}\right)_{h\in\mathcal{H}}$
as follows: $\mathbb{P}^{h}$ is a measure such that $X$ has density
$f_{X}^{\mathbb{P}^{h}}$ under $\mathbb{P}^{h}$, where $f_{X}^{\mathbb{P}^{h}}$
is obtained from $f_{X}^{\mathbb{Q}}$ by applying the Esscher transform.
The use of the Esscher transform can be supported by a utility maximizing
argument; see \cite{gerber1996actuarial}. In particular, we define
$\mathbb{P}^{h}$ such that 
\[
f_{X}^{\mathbb{P}^{h}}(x)=\frac{e^{hx}f_{X}^{\mathbb{Q}}(x)}{\int_{\mathbb{R}}e^{hy}f_{X}^{\mathbb{Q}}(y)dy}=\frac{e^{hx}f_{X}^{\mathbb{Q}}(x)}{E_{\mathbb{Q}}\left[S_{T}^{h}\right]},\quad x>0.
\]
It follows that 
\begin{equation*}
f_{S_{T}}^{\mathbb{P}^{h}}(x)=\frac{x^{h}}{x}\frac{f_{X}^{\mathbb{Q}}(\log(x))}{E_{\mathbb{Q}}\left[S_{T}^{h}\right]},\quad x>0.
\end{equation*}
The density $f_{S_{T}}^{\mathbb{P}^{h^{\ast}}}$ crosses $f_{S_{T}}^{\mathbb{P}^{h}}$
only once from above for $h^{\ast}<h$; hence, by \citet[Property 3.3.32]{denuit2005},
it follows that 
\begin{equation*}
F_{S_{T}}^{\mathbb{P}^{h^{\ast}}}\preceq_{\mathbb{F}_{FSD}}F_{S_{T}}^{\mathbb{P}^{h}}\Rightarrow F_{S_{T}}^{\mathbb{P}^{h^{\ast}}}\preceq_{\mathbb{F}_{SSD}}F_{S_{T}}^{\mathbb{P}^{h}},\quad h\in\mathcal{H}.
\end{equation*}
For the likelihood ratio it, holds that
\begin{equation*}
\ell^{\mathbb{P}^{h^{\ast}}}=\frac{f_{S_{T}}^{\mathbb{P}^{h^{\ast}}}(S_{T})}{f_{S_{T}}^{\mathbb{Q}}(S_{T})}=\frac{(S_{T})^{h^{\ast}}}{E_{\mathbb{Q}}\left[(S_{T})^{h^{\ast}}\right]},
\end{equation*}
which is strictly increasing in $S_{T}$ as $h^{\ast}>0$ and concave
if $h^{\ast}\in(0,1]$. We can apply Proposition~\ref{cor:for examples}
to show that conditions~\ref{assu:least_fav} and~\ref{assu:l*}
are satisfied for the sets $\mathbb{F}_{FSD}$ and $\mathbb{F}_{SSD}$
with least favorable measure $\Pa=\mathbb{P}^{h^{*}}$ and corresponding
likelihood ratio $\ell^{\ast}=\ell^{\P^{h^{*}}}.$ We can use Theorem
~\ref{thm:robust-cost-payoff} to compute the cost-efficient payoff
of a distribution function $F_{0}$.

\section{\label{sec:Robust-Portfolio-Selection}Robust portfolio selection}

\cite{gilboa1989maxmin} provide axioms that justify a maxmin expected
utility framework to make robust decisions when there is ambiguity on the
probability measure $\P$, i.e., when $\#{\mathcal{P}}>1.$
In this framework, \cite{schied2005optimal} shows that when a least
favorable measure $\Pa\in\mathcal{P}$ with respect to FSD ordering
(e.g., the stochastic integral ordering induced by the set $\mathbb{F}_{FSD}$
as defined in Section~\ref{sec:Robust-cost-efficiency}) exists, an
optimal portfolio can be derived. In this section, we extend the work
of \cite{schied2005optimal} in two different ways. First, we account for preferences beyond expected utility. Specifically, we derive optimal portfolios for robust preferences that are in accord with expected utility theory, rank dependent utility theory and Yaari's dual theory. Second, assuming the existence of a least
favorable measure $\Pa$ with respect to a general stochastic integral
ordering induced by some set $\mathbb{F}$, not necessarily identical to $\mathbb{F}_{FSD}$, we derive the optimal portfolio. Specifically, we derive optimal portfolios when a least favorable measure $\Pa\in\mathcal{P}$ with respect to SSD ordering exists (see Proposition~\ref{cor:for examples} for a sufficient condition) and the target distribution $F_0$ is sufficiently light tailed (see Remark~\ref{remark}).

\subsection{\label{subsec:Family-consistent-Preferences}Family consistent preferences}

A \emph{preference} $W$ is defined as a functional from the set of payoffs
$\mathcal{X}$ to the real line (\cite{he2017rank}, \cite{assa2018preferences}). %
Under preference $W$, the payoff $Y$ is preferred to $X$ if $W(X)\leq W(Y)$. In general,
$W$ may depend on the different measures $\mathbb{P}\in\mathcal{P}$
in a complicated way. In what follows, we denote a preference that
depends solely on some $\mathbb{P}\in\mathcal{P}$ by $W_{\mathbb{P}}$. 
\begin{definition}
\label{def:law_invariant}Let $\left(W_{\mathbb{P}}\right)_{\mathbb{P}\in\mathcal{P}}$
be a family of preferences. The preference $W_{\mathbb{P}}$, $\mathbb{P}\in\mathcal{P}$
is called \emph{$\mathbb{P}-$law invariant} if $F_{X}^{\mathbb{P}}=F_{Y}^{\mathbb{P}}$
implies that $W_{\mathbb{P}}(X)=W_{\mathbb{P}}(Y)$. The family $\left(W_{\mathbb{P}}\right)_{\mathbb{P}\in\mathcal{P}}$
is called \emph{law invariant} if each individual preference $W_{\mathbb{P}}$
is \emph{$\mathbb{P}-$}law invariant. 
\end{definition}


\begin{example}
A standard example of a $\P-$law invariant preference is $W_{\P}(X)=E_{\P}[u(X)]$
for some increasing utility function $u$. In this case, $W(X)=\underset{\mathbb{P}\in\mathcal{P}}{\inf}\,W_{\mathbb{P}}(X)$
amounts to the worst-case expected utility, commonly called \emph{robust
expected utility}, which was introduced in \cite{gilboa1989maxmin}. It is also referred to a \emph{robust utility functional} in \cite{schied2009robust}.
\end{example}

To the best of our knowledge, the next definition is new to the literature.
It will be helpful in solving robust portfolio choice problems.
\begin{definition}
\label{def:FSD-family consistent}Let $(W_{\mathbb{P}})_{\mathbb{P}\in\mathcal{P}}$
be a family of preferences. Let $\mathcal{Y}\subset\mathcal{X}$. The family of preferences $(W_{\mathbb{P}})_{\mathbb{P}\in\mathcal{P}}$
is called \emph{$\mathbb{F}-$family consistent on} $\mathcal{Y}$
\emph{with respect to} $\mathbb{P}^{\ast}\in\mathcal{P}$ if for all $Y\in\mathcal{Y}$
the inequality 
\begin{equation*}
F_{Y}^{\mathbb{P}^{\ast}}\preceq_{\mathbb{F}}F_{Y}^{\mathbb{P}},\quad\mathbb{P}\in\mathcal{P}
\end{equation*}
implies that
\[
W_{\mathbb{P}^{\ast}}(Y)\leq W_{\mathbb{P}}(Y),\quad\mathbb{P}\in\mathcal{P}.
\]
$\mathbb{F}-$family consistency of $(W_{\mathbb{P}})_{\mathbb{P}\in\mathcal{P}}$
with respect to some $\mathbb{P}^{\ast}\in\mathcal{P}$ has the following
interpretation: if a measure $\mathbb{P}^{\ast}$ yields the most
pessimistic view of any payoff $Y$ w.r.t. the stochastic ordering
induced by some set $\mathbb{F}$, then the preference under that
measure is the lowest as well.
\end{definition}

Next, we discuss some examples. Let $\mathcal{Y}\subset\mathcal{X}$ be a set of payoffs and let $\mathcal{D}$ be the set of cumulative distribution
functions induced by $\mathcal{Y}$, i.e.,
\[
\mathcal{D}=\{F_{Y}^{\mathbb{P}}:\,Y\in\mathcal{Y},\,\mathbb{P}\in\mathcal{P}\}.
\]
Let us consider an agent taking into account a family of law invariant preferences
$(W_{\mathbb{P}})_{\mathbb{P}\in\mathcal{P}},$  
i.e.,
\begin{equation}
W_{\mathbb{P}}(Y)=w(F_{Y}^{\mathbb{P}})\label{eq:W=00003Dw}
\end{equation}
 for some well defined $w:\mathcal{D}\rightarrow\mathbb{R}$. If $w$ respects integral
stochastic ordering, i.e., 
\begin{equation}
F\preceq_{\mathbb{F}}G\Rightarrow w(F)\leq w(G),\quad F,G\in\mathcal{D},\label{eq:wP(F) leq wP(G)}
\end{equation}
then $(W_{\mathbb{P}})_{\mathbb{P}\in\mathcal{P}}$ is $\mathbb{F}-$family
consistent on $\mathcal{Y}$ with respect to $\mathbb{P}^{\ast} \in \mathcal{P}$. 
We provide some specific examples in the contexts
of expected utility theory, Yaari's dual theory of choice and rank-dependent
expected utility theory:
\begin{example}
\label{exa:EUT_Yaari_RDU}Let $u:\mathbb{R}_+\to\mathbb{R}$. Let $\phi:[0,1]\rightarrow[0,1]$
with $\phi(0)=0$ and $\phi(1)=1$. For a given distribution function
$F$, define
\begin{align*}
w^{\text{EUT}}(F) & =\int_{\mathbb{R_+}}u(x)dF\\
w^{\text{Yaari}}(F) & =\int_{\mathbb{R_+}}\phi(1-F(x))dx\\
w^{\text{RDEU}}(F) & =\int_{\mathbb{R_+}}u(x)d(1-\phi(1-F(x)),
\end{align*}
where we tacitly assume that all integrals exist. It is straightforward to show that when $u$ and $\phi$ are non-decreasing, it holds that 
the family of preferences induced by $w^{\text{EUT}}$, $w^{\text{Yaari}}$
or $w^{\text{RDEU}}$ as in (\ref{eq:W=00003Dw}) is  \emph{
	$\mathbb{F}_{FSD}-$}family consistent on $\mathcal{Y} $ where $\mathcal{Y}$ is restricted to contain random variables such
that all relevant integrals exist. Furthermore, if $u$ is strictly increasing and concave and $\phi$
is strictly increasing, continuously differentiable and convex, we obtain that such a family is a \emph{
	$\mathbb{F}_{SSD}-$}family consistent on $\mathcal{Y}$; see   
\cite{yaari1987dual}, \cite{wang1998ordering},
\cite{he2017rank} and 
\cite{ryan2006risk}.

\end{example}

\begin{remark}
	
One could allow the function $w$ in (\ref{eq:W=00003Dw})
to depend on $\mathbb{P}$, i.e., define $W_{\mathbb{P}}(Y)=w_{\mathbb{P}}(F_{Y}^{\mathbb{P}})$
for some $w_{\mathbb{P}}:\mathcal{D}\rightarrow\mathbb{R}$, $\mathbb{P}\in\mathcal{P}$.
The family of preferences is then $\mathbb{F}-$family consistent
on $\mathcal{Y}$ with respect to $\mathbb{P}^{\ast}\in\mathcal{P}$
if both
\begin{equation*}
w_{\mathbb{P}^{\ast}}(F)\leq w_{\mathbb{P}}(F),\quad F\in\mathcal{D},\quad\mathbb{P}\in\mathcal{P}
\end{equation*}
and (\ref{eq:wP(F) leq wP(G)}) holds. To see this, let $Y\in\mathcal{Y}$
with $F_{Y}^{\mathbb{P}^{\ast}}\leq F_{Y}^{\mathbb{P}}$ for all $\mathbb{P}\in\mathcal{P}$.
It follows that 
\[
W_{\mathbb{P}^{\ast}}(Y)=w_{\mathbb{P}^{\ast}}\left(F_{Y}^{\mathbb{P}^{\ast}}\right)\leq w_{\mathbb{P}}\left(F_{Y}^{\mathbb{P}^{\ast}}\right)\leq w_{\mathbb{P}}\left(F_{Y}^{\mathbb{P}}\right)=W_{\mathbb{P}}(Y).
\]
\end{remark}

\subsection{\label{subsec:Optimal-portfolio-for robust preferences}Optimal portfolio
for robust preferences}

Inspired by \cite{gilboa1989maxmin} and \cite{schied2005optimal},
we consider the following problem:
\begin{problem}
\label{def:maximization_problem}Let $x_{0}>0$ be the initial wealth.
Let $\left(W_{\mathbb{P}}\right)_{\mathbb{P}\in\mathcal{P}}$ be a
family of preferences. We consider the robust maximization problem
\begin{equation}
\max_{X\in\mathcal{Y}_{(W_{\mathbb{P}})_{\mathbb{P}\in\mathcal{P}}}^{x_{0}}}\inf_{\mathbb{P}\in\mathcal{P}}W_{\mathbb{P}}(X),\label{eq:robust maximization problem}
\end{equation}
where $\mathcal{Y}_{(W_{\mathbb{P}})_{\mathbb{P}\in\mathcal{P}}}^{x_{0}}=\bigcap_{\mathbb{P}\in\mathcal{P}}\mathcal{Y}_{W_{\mathbb{P}}}^{x_{0}}$ and
\begin{equation}
\mathcal{Y}_{W_{\mathbb{P}}}^{x_{0}}:=\left\{ X\in\mathcal{X}\,:W_{\mathbb{P}}[X]\in\mathbb{R},\,e^{-rT}E_{\mathbb{Q}}\left[X\right]\leq x_{0}\right\} ,\quad\mathbb{P}\in\mathcal{P}.  
\end{equation}
\end{problem}

It turns out that under certain conditions a solution to the robust optimization problem \eqref{eq:robust maximization problem} can be found as a solution to a \emph{maximization problem under a single measure} $\mathbb{P}\in\mathcal{P}$.

\begin{problem} Let $x_{0}>0$ be the initial wealth.
	Let $W_{\mathbb{P}}$, ${\mathbb{P}\in\mathcal{P}}$ be a preference. We consider the \emph{maximization problem}
\label{def:maximization_problem2}

\begin{equation}
	\max_{X\in\mathcal{Y}_{W_{\mathbb{P}}}^{x_{0}}}W_{\mathbb{P}}(X).\label{eq:singleMaximization}
\end{equation}
\end{problem}

Under the assumption of the existence of a least favorable measure
with respect to FSD ordering, \cite{schied2005optimal} showed that
in order to solve the robust maximization problem \eqref{eq:robust maximization problem}
for preferences $\left(W_{\mathbb{P}}\right)_{\mathbb{P}\in\mathcal{P}}$,
${W}_{\P}(x)=E_{\P}[u(X)]$, it actually suffices to solve the single
measure maximization problem \eqref{eq:singleMaximization}. The following
theorem generalizes this result beyond the expected utility setting
to a general law invariant family of preferences  $\left(W_{\mathbb{P}}\right)_{\mathbb{P}\in\mathcal{P}}$.  
The theorem is illustrated in
Section~\ref{subsec:RDU+lognormal}, where we consider a robust rank-dependent
expected utility maximization problem for an investor with ambiguity on the trend and/or volatility of the risky asset.
\begin{thm}
\label{(Theorem-4.5).-Assume} Let $\mathbb{F}=\mathbb{F}_{FSD}$.
Given conditions  
\ref{assu:least_fav} and~\ref{assu:l*}, 
assume that $\left(W_{\mathbb{P}}\right)_{\mathbb{P}\in\mathcal{P}}$
is a law invariant and $\mathbb{F}_{FSD}-$consistent family of preferences
on $\mathcal{Y}_{(W_{\mathbb{P}})_{\mathbb{P}\in\mathcal{P}}}^{x_{0}}$
with respect to $\mathbb{P}^{\ast} \in\mathcal{P}$. Assume that the maximization problem
(\ref{eq:singleMaximization}) under $\mathbb{P}^{\ast}$ has a solution
$\tilde{X}\in\mathcal{Y}_{(W_{\mathbb{P}})_{\mathbb{P}\in\mathcal{P}}}^{x_{0}}.$
 Then it holds that 
\[
\underset{X\in\mathcal{Y}_{(W_{\mathbb{P}})_{\mathbb{P}\in\mathcal{P}}}^{x_{0}}}{\max}\,\underset{\mathbb{P}\in\mathcal{P}}{\inf}\,W_{\mathbb{P}}(X)=\underset{X\in\mathcal{Y}_{W_{\mathbb{P^{\ast}}}}^{x_{0}}}{\max}\,W_{\mathbb{P}^{\ast}}(X).
\]
\end{thm}

\begin{proof}
Let $h\in \mathbb{F}_{FSD}$ such that $h(\ell^{\ast})\in\mathcal{Y}_{(W_{\mathbb{P}})_{\mathbb{P}\in\mathcal{P}}}^{x_{0}}$.
Then, it holds by the $\mathbb{F}_{FSD}-$family consistency, condition
\ref{assu:least_fav} and (\ref{eq:F_X =00003D> F_f(X)})
that 
\begin{equation}
W_{\mathbb{P}^{\ast}}(h(\ell^{\ast}))\leq\underset{\mathbb{P}\in\mathcal{P}}{\inf}\,W_{\mathbb{P}}(h(\ell^{\ast})).\label{eq:A2}
\end{equation}
Let 
\[
X^{\ast}=\left[F_{\tilde{X}}^{\mathbb{P}^{\ast}}\right]^{-1}\left(F_{\ell^{\ast}}^{\mathbb{P}^{\ast}}(\ell^{\ast})\right).
\]
Then $X^{\ast}$ solves the standard cost-efficiency problem for $F_{\tilde{X}}^{\mathbb{P}^{\ast}}$
and thus $E_{\mathbb{Q}}[X^{\ast}]\leq E_{\mathbb{Q}}[\tilde{X}]$
and $F_{\tilde{X}}^{\mathbb{P}^{\ast}}=F_{X^{\ast}}^{\mathbb{P}^{\ast}}$;
hence, by the law invariance of $\left(W_{\mathbb{P}}\right)_{\mathbb{P}\in\mathcal{P}}$,
it holds that $X^{\ast}\in\mathcal{Y}_{(W_{\mathbb{P}})_{\mathbb{P}\in\mathcal{P}}}^{x_{0}}$.
It further holds that $X^{\ast}$ is a non-decreasing function of
$\ell^{\ast}$. It follows by (\ref{eq:A2}) that 
\begin{align*}
\underset{X\in\mathcal{Y}_{W_{\mathbb{P^{\ast}}}}^{x_{0}}}{\max}\,W_{\mathbb{P}^{\ast}}(X) & =W_{\mathbb{P}^{\ast}}(\tilde{X})=W_{\mathbb{P}^{\ast}}(X^{\ast})\\
 & \leq\underset{\mathbb{P}\in\mathcal{P}}{\inf}\,W_{\mathbb{P}}(X^{\ast})\\
 & \leq\underset{X\in\mathcal{Y}_{(W_{\mathbb{P}})_{\mathbb{P}\in\mathcal{P}}}^{x_{0}}}{\max}\,\underset{\mathbb{P}\in\mathcal{P}}{\inf}\,W_{\mathbb{P}}(X)\\
 & \leq\underset{X\in\mathcal{Y}_{(W_{\mathbb{P}})_{\mathbb{P}\in\mathcal{P}}}^{x_{0}}}{\max}\,W_{\mathbb{P}^{\ast}}(X)\\
 & \leq\underset{X\in\mathcal{Y}_{W_{\mathbb{P^{\ast}}}}^{x_{0}}}{\max}\,W_{\mathbb{P}^{\ast}}(X),
\end{align*}
where the last inequality follows by $\mathcal{Y}_{(W_{\mathbb{P}})_{\mathbb{P}\in\mathcal{P}}}^{x_{0}}\subset\mathcal{Y}_{W_{\mathbb{P^{\ast}}}}^{x_{0}}$. 
\end{proof}
From Theorem~\ref{(Theorem-4.5).-Assume}, it follows immediately
that solving robust preference maximization problems may reduce to solving an optimization problem
under a single probability measure. The following example illustrates
this consequence. 
\begin{example}
Assume $\mathbb{F}=\mathbb{F}_{FSD}$ and that conditions  
\ref{assu:least_fav} and~\ref{assu:l*} 
are satisfied.  
Let $W_{\mathbb{P}}(F)=w(F_{Y}^{\mathbb{P}})$ as in (\ref{eq:W=00003Dw}),
where $w\in\{w^{\text{EUT}},w^{\text{Yaari}},w^{\text{RDEU}}\}$ as in Example~\ref{exa:EUT_Yaari_RDU}. Assuming a solution to \eqref{eq:singleMaximization} under $\mathbb{P}^\ast \in\mathcal{P}$ exists, then it follows that
\[
\underset{X\in\mathcal{Y}_{(W_{\mathbb{P}})_{\mathbb{P}\in\mathcal{P}}}^{x_{0}}}{\max}\,\underset{\mathbb{P}\in\mathcal{P}}{\inf}\,W_{\mathbb{P}}(X)=\underset{X\in\mathcal{Y}_{W_{\mathbb{P^{\ast}}}}^{x_{0}}}{\max}\,W_{\mathbb{P}^{\ast}}(X).
\]
\end{example}

The main assumption in Theorem~\ref{(Theorem-4.5).-Assume} that is needed to solve the robust maximization problem \eqref{eq:robust maximization problem} in the case of a family of law invariant preferences $\left(W_{\mathbb{P}}\right)_{\mathbb{P}\in\mathcal{P}}$ is the existence of a least favorable measure $\mathbb{P}^\ast$ with respect to $\mathbb{F}_{FSD}$. In the following theorem we show that it is possible to weaken this assumption in that we only require existence of a least favorable measure $\mathbb{P}^\ast$ with respect to some $\mathbb{F}\subset \mathbb{F}_{FSD}$, e.g., $\mathbb{F}=\mathbb{F}_{SSD}$. The theorem is illustrated in
Section~\ref{subsec:RDU+lognormal}, where we consider a robust rank-dependent
expected utility maximization problem for an investor who faces ambiguity on expected return and volatility of the risky asset.

\begin{thm}
\label{thm:SSD_port}Consider a given set $\mathbb{F}.$ Given conditions 
\ref{assu:least_fav} and~\ref{assu:l*},
assume that the maximization problem (\ref{eq:singleMaximization})
under $\mathbb{P}^{\ast}$ has a solution $\tilde{X}\in\mathcal{Y}_{(W_{\mathbb{P}})_{\mathbb{P}\in\mathcal{P}}}^{x_{0}}$,
which can $\mathbb{P}^{\ast}-$a.s. be expressed as $f(\ell^{\ast})$
for some $f\in\mathbb{F}$. Further, assume that $\left(W_{\mathbb{P}}\right)_{\mathbb{P}\in\mathcal{P}}$
is $\mathbb{F}-$family consistent on $\mathcal{Y}_{(W_{\mathbb{P}})_{\mathbb{P}\in\mathcal{P}}}^{x_{0}}$
with respect to $\mathbb{P}^{\ast}$. Then it holds that 
\[
\underset{X\in\mathcal{Y}_{(W_{\mathbb{P}})_{\mathbb{P}\in\mathcal{P}}}^{x_{0}}}{\max}\,\underset{\mathbb{P}\in\mathcal{P}}{\inf}\,W_{\mathbb{P}}(X)=\underset{X\in\mathcal{Y}_{W_{\mathbb{P^{\ast}}}}^{x_{0}}}{\max}\,W_{\mathbb{P}^{\ast}}(X).
\]
\end{thm}

\begin{proof}
Let $h\in\mathbb{F}$ such that $h(\ell^{\ast})\in\mathcal{Y}_{(W_{\mathbb{P}})_{\mathbb{P}\in\mathcal{P}}}^{x_{0}}$.
Then, (\ref{eq:A2}) holds true by the $\mathbb{F}-$family
consistency and (\ref{eq:F_X =00003D> F_f(X)}). By assumption,
$\tilde{X}=f(\ell^{\ast})$, $\mathbb{P}^{\ast}-$a.s. for some $f\in\mathbb{F}$.
Hence, it holds that 
\begin{align*}
\underset{X\in\mathcal{Y}_{W_{\mathbb{P^{\ast}}}}^{x_{0}}}{\max}\,W_{\mathbb{P}^{\ast}}(X) & =W_{\mathbb{P}^{\ast}}(\tilde{X})\\
 & \leq\underset{\mathbb{P}\in\mathcal{P}}{\inf}\,W_{\mathbb{P}}(\tilde{X})\\
 & \leq\underset{X\in\mathcal{Y}_{(W_{\mathbb{P}})_{\mathbb{P}\in\mathcal{P}}}^{x_{0}}}{\max}\,\underset{\mathbb{P}\in\mathcal{P}}{\inf}\,W_{\mathbb{P}}(X)\\
 & \leq\underset{X\in\mathcal{Y}_{(W_{\mathbb{P}})_{\mathbb{P}\in\mathcal{P}}}^{x_{0}}}{\max}\,W_{\mathbb{P}^{\ast}}(X)\\
 & \leq\underset{X\in\mathcal{Y}_{W_{\mathbb{P^{\ast}}}}^{x_{0}}}{\max}\,W_{\mathbb{P}^{\ast}}(X),
\end{align*}
in which the last inequality follows by $\mathcal{Y}_{(W_{\mathbb{P}})_{\mathbb{P}\in\mathcal{P}}}^{x_{0}}\subset\mathcal{Y}_{W_{\mathbb{P^{\ast}}}}^{x_{0}}$. 
\end{proof}

Note that, as compared to the statements in   Theorem~\ref{(Theorem-4.5).-Assume}, in 
Theorem \ref{thm:SSD_port} the $W_{\mathbb{P}}$ do not need to be law-invariant. Moreover, as long as the solution can be expressed as a certain function of the likelihood ratio $\ell^{\ast}$,  
the preferences do not need to be increasing, i.e., $X\leq Y$
does not need to imply $W_{\mathbb{P}}(X)\leq W_{\mathbb{P}}(Y).$
As pointed out, Theorem \ref{thm:SSD_port} is applicable, in particular, for the case $\mathbb{F}=\mathbb{F}_{SSD}$. 
However, the requirement that $\tilde{X}\in\mathcal{Y}_{(W_{\mathbb{P}})_{\mathbb{P}\in\mathcal{P}}}^{x_{0}}$
can be $\mathbb{P}^{\ast}-$a.s. expressed as $h(\ell^{\ast})$
for some $h\in\mathbb{F}_{SSD}$ is equivalent to $h$ being increasing and concave. This property is difficult to verify ex-ante. Hereafter, we show that in the case in which  $W_{\P}$ is an expected utility, this condition
translates into an easy-to-verify condition on the utility function. We formulate the following theorem. 
\begin{thm}
	\label{thm:SSD_port_utility}Let $\mathbb{P}\in\mathcal{P}$ with
	likelihood ratio $\ell^{\mathbb{P}}$.
	Let $u:\mathbb{R}_+\rightarrow\mathbb{R}$ be a differentiable, concave and
	strictly increasing utility function such that $u^{\prime}$ is strictly
	decreasing. If the maximization
	problem (\ref{eq:singleMaximization}) under $\mathbb{P}$ has a solution,
	then the solution is a non-decreasing and concave function of $\ell^{\mathbb{P}}$
	if and only if $\frac{1}{u^{\prime}}$ is convex. If $u$ is three
	times differentiable, $\frac{1}{u^{\prime}}$ is convex if and only
	if 
	\begin{equation}
		a(x)\geq\frac{p(x)}{2},\label{absolute}
	\end{equation}
	in which $a(x):=-\frac{u^{\prime\prime}(x)}{u^{\prime}(x)}$ refers
	to the absolute risk aversion measure and $p(x):=-\frac{u^{\prime\prime\prime}(x)}{u^{\prime\prime}(x)}$
	to the absolute prudence. 
\end{thm}

\begin{proof}
	By Lemma 2 in \cite{bernard2015rationalizing}, the solution to (\ref{eq:singleMaximization}) is unique
	and given by $[u^{\prime}]^{-1}\left(\frac{c_{0}}{\ell^{\mathbb{P}}}\right)$
	for some $c_{0}>0$. See also \cite{merton1975optimum} for a proof
	in a context in which Inada's conditions are satisfied. Note that
	$u^{\prime}>0$ and that $\frac{1}{u^{\prime}}$ is strictly increasing.
	Observe that the inverse of $\frac{1}{u^{\prime}}$ is $x\mapsto[u^{\prime}]^{-1}(\frac{1}{x})$,
	which is hence also strictly increasing. The inverse of a convex (concave)
	and strictly increasing function is concave (convex). For the second
	assertion, observe that $u^{\prime\prime}<0$ and that a function
	is convex on an open interval if and only if its second derivative
	is non-negative. Then \eqref{absolute} follows immediately.
\end{proof}
\begin{remark}
	\cite{maggi2006relationship} have shown that $a(x)>p(x)$ if and
	only if the utility has increasing absolute risk aversion, which is
	somewhat unusual (it is typically assumed that agents have decreasing absolute risk aversion
	given that they become less risk averse as their wealth increases). Here,
	our condition \eqref{absolute} is not incompatible with decreasing absolute risk aversion
	due to the factor $\frac{1}{2}$.
	
	\vspace{3pt}
	Condition \eqref{absolute} has appeared
	several times in the literature. It has been found to play a role
	in the context of insurance models in \cite{bourles2017prevention},
	but it also appeared as a condition in the opening of a new asset
	market (\cite{gollier1996toward}), when there is uncertainty on the
	size (\cite{gollier2000scientific}) or the probability of losses
	(\cite{gollier2002optimal}) and under contingent auditing (\cite{sinclair1997environmental}).
	Further interpretation of this condition and, in particular, of the
	degree of concavity of the inverse of the marginal utility can be
	found in \cite{bourles2017prevention}. This condition also appears
	in \cite{varian1985divergence} in the context of portfolio selection
	under ambiguity. 
\end{remark}

\begin{example}\label{CRRA}
	As an illustration of Theorem~\ref{thm:SSD_port_utility}, we provide two utility functions, which are differentiable,
	concave and strictly increasing functions such that one over the marginal
	utility is convex. 
	\begin{itemize}
		\item The exponential utility for risk-averse agents: $u:\mathbb{R}_+\rightarrow\mathbb{R},\,x\mapsto1-e^{-\lambda x},$
		for $\lambda>0$. It holds that $\frac{1}{u^{\prime}(x)}=e^{\lambda x}$,
		which is strictly increasing and convex. 
		\item CRRA utility: $u:\mathbb{R}_+\rightarrow\mathbb{R},\,x\mapsto\frac{x^{1-\eta}}{1-\eta},$
		for $\eta>1$. It holds that $\frac{1}{u^{\prime}(x)}=x^{\eta}$,
		which is strictly increasing and convex. 
	\end{itemize}
\end{example}

\subsection{\label{subsec:RDU+lognormal}Rank-dependent utility in log-normal
markets}

We now discuss some examples to illustrate Section~\ref{subsec:Optimal-portfolio-for robust preferences}
in a log-normal market setting with uncertainty on the drift and volatility.
In particular, we explicitly solve a robust rank-dependent expected
utility problem using Theorems~\ref{(Theorem-4.5).-Assume} and~\ref{thm:SSD_port}.
As in Section~\ref{subsec:Examples}, we assume that the real-world
distribution of $S_{T}$ is log-normal with parameters $\log(S_{0})+(\mu-\frac{\sigma^{2}}{2})T$
and $\sigma^{2}T$ and that the investor has uncertainty on the trend
and potentially also the volatility. She may expect the true parameters to lie within the
cube 
\[
\mathcal{D}^{\mu_{1},\mu_{2},\sigma_{1},s}=\left\{ (\mu,\sigma)\subset\mathbb{R}^{2}\,:\,\mu_{1}\leq\mu\leq\mu_{2},\,\sigma_{1}\leq\sigma\leq s\right\} 
\]
for $r<\mu_{1}<\mu_{2}$ and $0<\sigma_{1}\leq s$. The investor thus
considers $\mathcal{P}=\left(\mathbb{P}^{\mu,\sigma}\right)_{(\mu,\sigma)\in\mathcal{D}^{\mu_{1},\mu_{2},\sigma_{1},s}}$
as the set of all plausible probability measures on $(\Omega,\mathcal{F})$. Note that if $\sigma_1=s,$ the investor only faces drift ambiguity; otherwise, she considers ambiguity on both trend and volatility. 
Under $\P^{\mu,\sigma}$, $S_{T}$ is log-normal with density $f^{\mu,\sigma}$,
defined in Equation (\ref{eq:f^m,sig}). In the next example, $\Phi$
is the standard normal distribution function.
\begin{example}
\label{exa:RDU}Let $U(x)=\frac{x^{1-\eta}}{1-\eta}$, $x\geq0$ for
$\eta\in(0,1)$ be the CRRA utility function. Let $\gamma\in\mathbb{R}$
and let $w(u)=\Phi(\Phi^{-1}(u)+\gamma)$, $u\in[0,1]$ denote the
so-called Wang transform, which is increasing concave if $\gamma>0$
and increasing convex if $\gamma<0$. Consider the following portfolio
choice problem; in which the investor maximizes her expected rank-dependent
utility: 
\begin{equation}
\max_{X\in\mathcal{Y}_{W_{\mathbb{P}^{\mu,\sigma}}}^{x_{0}}}\int_{0}^{\infty}U(x)d(1-w(1-F_{X}^{\mathbb{P}^{\mu,\sigma}}(x)))\label{eq:solution_standard_RDU}
\end{equation}
where $x_{0}>0$ is the initial wealth and $(\mu,\sigma)\in\mathcal{D}^{\mu_{1},\mu_{2},\sigma_{1},s}$.
Let $\theta:=\sqrt{T}\frac{\mu-r}{\sigma}$. The solution to (\ref{eq:solution_standard_RDU})
is given by 
\begin{equation}
X_{\mu,\sigma}^{\ast}:=\begin{cases}
\lambda^{-\frac{1}{\eta}}\exp\left(\frac{rT}{\eta}-\frac{1}{2}\frac{\gamma}{\eta}(\theta+\gamma)\right)\left(\ell^{\mathbb{P}}\right)^{\frac{\gamma}{\theta\eta}+\frac{1}{\eta}} & ,\gamma>-\theta\\
\lambda^{-\frac{1}{\eta}} & \text{, otherwise,}
\end{cases}\label{eq:solution_RDU}
\end{equation}
where $\lambda$ depends on $\eta$, $\gamma$ and $\theta$, see
Equations (\ref{eq:lam1}), (\ref{eq:lam2}) and (\ref{eq:lam3}).
The solution to the robust rank-dependent utility problem 
\begin{equation}
\max_{X\in\mathcal{Y}_{(W_{\mathbb{P}})_{\mathbb{P}\in\mathcal{P}}}^{x_{0}}}\inf_{\mathbb{P}\in\mathcal{P}}\int_{0}^{\infty}U(x)d(1-w(1-F_{X}^{\mathbb{P}}(x)))dx\label{eq:robust_RDU}
\end{equation}
is given by $X_{\mu_{1},s}^{\ast}$ if there is no ambiguity on the
volatility, i.e., when $\sigma_{1}=s$. If there is ambiguity on the
volatility, $\gamma<0$ and $\frac{\mu_{1}-r}{s^{2}}\in(0,1]$, then
$X_{\mu_{1},s}^{\ast}$ still solves (\ref{eq:robust_RDU}). 
\end{example}

\begin{proof}
We first prove that $X_{\mu,\sigma}^{\ast}$ solves (\ref{eq:solution_standard_RDU}).
Let $\mathbb{P}=\mathbb{P}^{\mu,\sigma}$. The state price $\xi^{\mathbb{P}}:=\frac{e^{-rT}}{\ell^{\mathbb{P}}}$
is log-normally distributed with parameters $-rT-\frac{1}{2}\theta^{2}$
and $\theta>0$, see Equation (\ref{eq:h^mu,s}). Hence, as $\mu>r$,
it holds that 
\begin{align*}
F_{\xi^{\mathbb{P}}}^{\mathbb{P}}(x) & =\mathbb{P}\left(\xi^{\mathbb{P}}\leq x\right)=\Phi\left(\frac{\log(x)+rT+\frac{1}{2}\theta^{2}}{\theta}\right),\quad x>0
\end{align*}
and 
\[
\left[F_{\xi^{\mathbb{P}}}^{\mathbb{P}}\right]^{-1}(p)=\exp\left(\Phi^{-1}(p)\theta-rT-\frac{1}{2}\theta^{2}\right),\quad p\in(0,1).
\]
Let 
\[
H(z)=-\int_{0}^{w^{-1}(1-z)}\left[F_{\xi^{\mathbb{P}}}^{\mathbb{P}}\right]^{-1}(t)dt,\quad z\in[0,1].
\]
The solution to the classical rank-utility problem (\ref{eq:solution_standard_RDU})
is well-known (see for instance Theorem 4.1 in \cite{xu2016note} or Section 3.2 in \cite{ruschendorf2020construction}),
and is given by 
\[
X_{\mu,\sigma}^{\ast}=[U^{\prime}]^{-1}\left(\lambda\hat{H}^{\prime}\left(1-w\left(F_{\xi^{\mathbb{P}}}^{\mathbb{P}}\left(\xi^{\mathbb{P}}\right)\right)\right)\right),
\]
where $\lambda$ is determined by $E_{\mathbb{P}}\left[\xi^{\mathbb{P}}X_{\mu,\sigma}^{\ast}\right]=x_{0}$
and $\hat{H}$ is the concave envelope of $H$. Using $w^{\prime}(u)=\frac{\Phi^{\prime}(\Phi^{-1}(u)+\gamma)}{\Phi^{\prime}(\Phi^{-1}(u))}$
, after some calculations, we obtain that 
\begin{align*}
H^{\prime}(z) & =\frac{\left[F_{\xi^{\mathbb{P}}}^{\mathbb{P}}\right]^{-1}(w^{-1}(1-z))}{w^{\prime}(w^{-1}(1-z))}=\exp\left(\Phi^{-1}(1-z)(\gamma+\theta)-rT-\frac{1}{2}(\gamma+\theta)^{2}\right).
\end{align*}
We distinguish two cases to find a more explicit expression for $X_{\mu,\sigma}^{\ast}$.
\emph{Case 1}: $\gamma+\theta>0$. Then, $H^{\prime}$ is non-increasing and
hence $H$ is concave and equal to $\hat{H}$. As $[U^{\prime}]^{-1}(y)=y^{-\frac{1}{\eta}}$,
it is easy to see that 
\begin{align*}
X_{\mu,\sigma}^{\ast} & =\lambda^{-\frac{1}{\eta}}\exp\left(-\frac{rT\gamma}{\theta\eta}-\frac{1}{2}\frac{\gamma}{\eta}(\theta+\gamma)\right)\left(\xi^{\mathbb{P}}\right)^{-\frac{\gamma}{\theta\eta}-\frac{1}{\eta}}\\
 & =\lambda^{-\frac{1}{\eta}}\exp\left(\frac{rT}{\eta}-\frac{1}{2}\frac{\gamma}{\eta}(\theta+\gamma)\right)\left(\ell^{\mathbb{P}}\right)^{\frac{\gamma}{\theta\eta}+\frac{1}{\eta}}.
\end{align*}
If $1-\frac{\gamma}{\theta\eta}-\frac{1}{\eta}=0$, then $\xi^{\mathbb{P}}X_{\mu,\sigma}^{\ast}$
is constant and it holds that
\begin{equation}
\lambda=x_{0}^{-\eta}\exp\left(-\frac{rT\gamma}{\theta}-\frac{1}{2}\gamma(\theta+\gamma)\right).\label{eq:lam1}
\end{equation}
Otherwise, $\xi^{\mathbb{P}}X_{\mu,\sigma}^{\ast}$ is log-normally
distributed%
{} and it follows that 
\begin{equation}
\lambda=x_{0}^{-\eta}\exp\left(rT(1-\eta)+\frac{1}{2}\left(\theta^{2}(1-\eta)+\gamma^{2}+(\theta\eta-(\gamma+\theta))^{2}\right)\right).\label{eq:lam2}
\end{equation}
\emph{Case 2}: Assume $\gamma+\theta\leq0$. Then, $H$ is convex.
Note that $H(0)=-1$ and $H(1)=0$. As $H$ is convex, the concave
envelope $\hat{H}$ of $H$ is given by $\hat{H}(x)=x-1$. Then $\hat{H}^{\prime}\equiv1$
and 
\[
X_{\mu,\sigma}^{\ast}=[U^{\prime}]^{-1}\left(\lambda\right)=\lambda^{-\frac{1}{\eta}}.
\]
Therefore, $\xi^{\mathbb{P}}X_{\mu,\sigma}^{\ast}=\lambda^{-\frac{1}{\eta}}\xi^{\mathbb{P}}$,
and hence
\begin{equation}
\lambda=x_{0}^{-\eta}e^{-rT\eta}.\label{eq:lam3}
\end{equation}
Assume that there is no ambiguity on the volatility, i.e., $\sigma_{1}=s$.
Section \ref{subsec:Example:-Uncertainty-about-drift} shows that
the least favorable measure with respect to $\mathbb{F}_{FSD}$ is
given by $\mathbb{P}^{\ast}=\mathbb{P}^{\mu_{1},s}$ with corresponding
likelihood ratio $\ell^{\ast}=\ell^{\P^{\mu_{1},s}}$. By Example
\ref{exa:EUT_Yaari_RDU}, the preference in (\ref{eq:robust_RDU})
is $\mathbb{F}_{FSD}-$family consistent and Theorem~\ref{(Theorem-4.5).-Assume}
shows that $X_{\mu,\sigma}^{\ast}$ solves the robust rank-dependent
utility problem (\ref{eq:robust_RDU}). Lastly, assume that there
is ambiguity on the volatility. If $\frac{\mu_{1}-r}{s^{2}}\in(0,1]$,
Proposition \ref{cor:log_normal} shows that the least favorable measure
with respect to $\mathbb{F}_{SSD}$ is also $\mathbb{P}^{\ast}$.
If $\gamma<0$, $X_{\mu,\sigma}^{\ast}$ is a concave and non-decreasing
function of $\ell^{\ast}$. By Example~\ref{exa:EUT_Yaari_RDU},
the preference in (\ref{eq:robust_RDU}) is $\mathbb{F}_{SSD}-$family
consistent. Apply Theorem~\ref{thm:SSD_port} to show that also in
this case, $X_{\mu,\sigma}^{\ast}$ solves the robust rank-dependent
utility problem (\ref{eq:robust_RDU}).
\end{proof}
To better understand the solution in Example~\ref{exa:RDU}, we ``rationalize''
the solution as in \cite{bernard2015rationalizing}, i.e., we show
that the optimal investment strategy in the robust rank-dependent
setting also solves an expected utility maximization problem. Example \ref{exa:EUT} shows that the solution to the expected rank-dependent
utility problems (\ref{eq:solution_standard_RDU}) and (\ref{eq:robust_RDU})
involving a Wang transform with parameter $\gamma$ and a CRRA utility
function with parameter $\eta$ can be rationalized by a CRRA utility
with parameter $\frac{\eta\theta}{\gamma+\theta}$.
\begin{example}
\label{exa:EUT}Let $(\mu,\sigma)\in\mathcal{D}^{\mu_{1},\mu_{2},\sigma_{1},s}$,
$X_{\mu,\sigma}^{\ast}$, $\theta$, $x_{0}$ and $\mathcal{Y}_{W_{\mathbb{P}^{\mu,\sigma}}}^{x_{0}}$
as in Example~\ref{exa:RDU}, such that $\gamma>-\theta$
and $\eta\theta\neq\gamma+\theta$. $X_{\mu,\sigma}^{\ast}$ solves
the following expected utility maximization problem:
\[
\max_{X\in\mathcal{Y}_{W_{\mathbb{P}^{\mu,\sigma}}}^{x_{0}}}\int_{0}^{\infty}u(x)dF_{X}^{\mathbb{P}^{\mu,\sigma}}(x)
\]
for the utility function 
\begin{equation}
u(x)=\frac{1}{1-\frac{\eta\theta}{\gamma+\theta}}x^{1-\frac{\eta\theta}{\gamma+\theta}}.\label{eq:u}
\end{equation}
The function $u:\mathbb{R}_{+}\to\mathbb{R}$ is non-decreasing and
concave.

\end{example}

\begin{proof}
Note that $\gamma>-\theta$ implies $\frac{\gamma}{\theta\eta}+\frac{1}{\eta}\neq0$.
Let $\mathbb{P}=\mathbb{P}^{\mu,\sigma}$. Let $\xi^{\mathbb{P}}:=\frac{e^{-rT}}{\ell^{\mathbb{P}}}$.
As in \cite{bernard2015rationalizing}, let $c>0$ and define 
\begin{equation*}
\tilde{u}(x)=\int_{c}^{x}\left[F_{\xi^{\mathbb{P}}}^{\mathbb{P}}\right]^{-1}\left(1-F_{X_{\mu,\sigma}^{\ast}}^{\mathbb{P}}(y)\right)dy,\quad x\in\mathbb{R}_{+}.
\end{equation*}
$X_{\mu,\sigma}^{\ast}$ is log-normally distributed.%
{} It follows that 
\begin{align*}
\left[F_{\xi^{\mathbb{P}}}^{\mathbb{P}}\right]^{-1}\left(1-F_{X_{\mu,\sigma}^{\ast}}^{\mathbb{P}}(x)\right) & =\kappa x^{-\frac{\eta\theta}{\gamma+\theta}},
\end{align*}
where $\kappa>0$ is a suitable constant. Thus,
\begin{align*}
\tilde{u}(x) & =\kappa\frac{\gamma+\theta}{\theta(1-\eta)+\gamma}\left(x^{-\frac{\eta\theta}{\gamma+\theta}+1}-c^{-\frac{\eta\theta}{\gamma+\theta}+1}\right).
\end{align*}
In summary, as $\tilde{u}$ is only determined up to positive affine
transformations, $X_{\mu,\sigma}^{\ast}$ solves the expected utility
maximization problem for the utility given in (\ref{eq:u}): see Theorem
2 in \cite{bernard2015rationalizing}.  
\end{proof}
\section{\label{sec:Rationalizing-robust-cost-effici}Rationalizing robust
cost-efficient payoffs}

When there is no ambiguity on the probability measure $\mathbb{P}$,
there is a close relationship between cost-efficiency and portfolio optimization:
for any cost-efficient payoff $X$, there exists a utility
function $u$ (unique up to a linear transformation) such that $X$
also solves the expected utility maximization problem (\cite{bernard2015rationalizing}).
In this section, we show that this result can be generalized to the
robust setting developed previously in that robust cost-efficient
payoffs can be rationalized in terms of the maxmin utility framework
introduced in \cite{gilboa1989maxmin}. Specifically, we show $-$
under the same assumptions as in Theorems~\ref{(Theorem-4.5).-Assume}
and~\ref{thm:SSD_port} $-$ that payoffs maximize a robust utility
functional as in \cite{gilboa1989maxmin} if and only if they are
robust cost-efficient.

In the following theorem we distinguish two cases: a) we deal with
law-invariant preferences and FSD ordering and assume that the various
(robust) maximization problems have unique solutions. Or, b) we deal with 
general preferences  and stochastic ordering and we do not require
uniqueness of the solutions but assume that the solution $X^{\ast}$
of the various maximization problems can be written as $X^{\ast}=f(\ell^{\ast})$
for some $f\in\mathbb{F}$. 
\begin{thm}
\label{thm:ratio_SSD}Assume $\mathbb{F}_{SSD}\subset\mathbb{F}\subset\mathbb{F}_{FSD}$.
Let conditions~\ref{assu:setF}, 
\ref{assu:least_fav} and~\ref{assu:l*}  
hold. Let $X^{\ast}\in \mathcal{X}$ be a payoff. 
Let $x_{0}=e^{-rT}E_{\mathbb{Q}}[X^{\ast}]<\infty$. Let $c>0$, such that $F_{X^{\ast}}^{\mathbb{P}^{\ast}}(c)>0,$ $\mathbb{P}^{\ast}\in\mathcal{P}$.
Let $\xi^{\ast}=\frac{e^{-rT}}{\ell^{\ast}}$ and define  
\begin{equation}
u(x)=\int_{c}^{x}F_{\xi^{\ast}}^{-1}\left(1-F_{X^{\ast}}^{\mathbb{P}^{\ast}}(y)\right)dy,\quad x\in\mathbb{R}_+.\label{eq:u-1}
\end{equation}
Assume that $E_{\mathbb{P}}[u(X^{\ast})]<\infty$ for all $\PinP$.
We further assume that one of the following two conditions is satisfied: 
\begin{description}
\item [{\emph{(a)}}] $\mathbb{F}=\mathbb{F}_{FSD}$ 
\item [{\emph{(b)}}] $X^{\ast}=f(\ell^{\ast})$, $\mathbb{P}^{\ast}-$a.s.
for some $f\in\mathbb{F}$ and $\left[F_{X^{\ast}}^{\mathbb{P^{\ast}}}\right]^{-1}\circ F_{\ell^{\ast}}^{\mathbb{P}^{\ast}}\in\mathbb{F}$. 
\end{description}
Then, the following statements are equivalent: 
\end{thm}

\begin{description}
\item [{\emph{i)}}] \emph{$X^{\ast}$ is cost-efficient under $\mathbb{P}^{\ast}$.} 
\item [{\emph{ii)}}] \emph{It holds that $X^{\ast}=\left[F_{X^{\ast}}^{\mathbb{P^{\ast}}}\right]^{-1}\left(F_{\ell^{\mathbb{P^{\ast}}}}^{\mathbb{P^{\ast}}}\left(\ell^{\mathbb{P}^{\ast}}\right)\right)$,
$\mathbb{P}^{\ast}-$a.s.} 
\item [{\emph{iii)}}] \emph{$X^{\ast}$ is $\mathbb{P}^{\ast}-$a.s. non-decreasing
in $\ell^{\ast}$.} 
\item [{\emph{iv)}}] \emph{$X^{\ast}$ solves the $\mathbb{F}-$robust
cost-efficiency problem for $F_{X^{\ast}}^{\mathbb{P}^{\ast}}$.} 
\item [{\emph{v)}}] \emph{$X^{\ast}$ solves the expected utility maximization
problem under $\mathbb{P}^{\ast}$ for the utility function
$u$ 
\[
\underset{X\in\mathcal{Y}_{E_{\mathbb{P}^{\ast}}[u(.)]}^{x_{0}}}{\max}\,E_{\mathbb{P}^{\ast}}[u(X)].
\]
} 
\item [{\emph{vi)}}] \emph{$X^{\ast}$ solves the robust expected utility
problem for the utility function $u$ 
\[
\underset{X\in\mathcal{Y}_{(E_{\mathbb{P}}[u(.)])_{\mathbb{P}\in\mathcal{P}}}^{x_{0}}}{\max}\,\underset{\mathbb{P}\in\mathcal{P}}{\inf}\,E_{\mathbb{P}}[u(X)]
\]
and the solution is unique if condition (a) is satisfied.} 
\item [{\emph{vii)}}] \emph{There is a family of preferences $\left(W_{\mathbb{P}}\right)_{\mathbb{P}\in\mathcal{P}}$ that is $\mathbb{F}-$family consistent on $\mathcal{Y}_{(W_{\mathbb{P}})_{\mathbb{P}\in\mathcal{P}}}^{x_{0}}$
with respect to} $\mathbb{P}^{\ast}$\emph{ such that $X^{\ast}\in\mathcal{Y}_{(W_{\mathbb{P}})_{\mathbb{P}\in\mathcal{P}}}^{x_{0}}$
and $X^{\ast}$ is the solution to the maximization problem under
$\mathbb{P}^{\ast}$ 
\[
\underset{X\in\mathcal{Y}_{W_{\mathbb{P^{\ast}}}}^{x_{0}}}{\max}\,W_{\mathbb{P}^{\ast}}(X)
\]
and the solution is unique and the family of preferences is law invariant
if condition (a) is satisfied.} 
\item [{\emph{viii)}}] \emph{There is a family of preferences $\left(W_{\mathbb{P}}\right)_{\mathbb{P}\in\mathcal{P}}$ that is $\mathbb{F}-$family consistent on $\mathcal{Y}_{(W_{\mathbb{P}})_{\mathbb{P}\in\mathcal{P}}}^{x_{0}}$
with respect to} $\mathbb{P}^{\ast}$\emph{ such that $X^{\ast}$
is the solution to the robust maximization problem 
\[
\underset{X\in\mathcal{Y}_{(W_{\mathbb{P}})_{\mathbb{P}\in\mathcal{P}}}^{x_{0}}}{\max}\,\underset{\mathbb{P}\in\mathcal{P}}{\inf}\,W_{\mathbb{P}}(X)
\]
and the solution is unique and the family of preferences is law invariant
if condition (a) is satisfied.} 
\end{description}
\begin{proof}
The equivalence between i), ii) and iii) follows from Lemma~\ref{lem:N_cost_eff_P}.
The equivalence between iv) and ii) follows from Theorem~\ref{thm:robust-cost-payoff}
and Remark~\ref{rem:Instead-of-requiring}. Note that $\left[F_{X^{\ast}}^{\mathbb{P^{\ast}}}\right]^{-1}\circ F_{\ell^{\ast}}^{\mathbb{P}^{\ast}}\in\mathbb{F}_{FSD}$
is always true. By Theorem 3 in \cite{bernard2015rationalizing},
i) implies v). By Lemma~\ref{cor:unique_P*} and Lemma 3 in \cite{bernard2015rationalizing},
v) implies i). By Example \ref{exa:EUT_Yaari_RDU}, v) implies vii) trivially,
define $W_{\mathbb{P}}(.)=E_{\mathbb{P}^{\ast}}[u(.)]$ for all $\mathbb{P}\in\mathcal{P}$.
v) implies vi), if (a) holds, by Theorem~\ref{(Theorem-4.5).-Assume}
and Example \ref{exa:EUT_Yaari_RDU} as $E_{\mathbb{P}}[u(X^{\ast})]<\infty$
for all $\PinP$ by assumption. v) implies vi), if (b) holds, by Theorem
\ref{thm:SSD_port}. vi) implies viii) trivially. By Lemma \ref{cor:unique_P*}
and Lemma~\ref{cor:unique_robust}, if (a) holds, vi)$\Rightarrow$i) and
vii)$\Rightarrow$i) and viii)$\Rightarrow$i). Note that, if (b)
holds, iii) is always true because $f\in\mathbb{F}\subset\mathbb{F}_{FSD}$
is non-decreasing. 
\end{proof}

\begin{remark}
Assuming that all functions in $\mathbb{F}$ are non-decreasing,
i.e., that $\mathbb{F}\subset\mathbb{F}_{FSD}$, is not really a restriction. Otherwise, there
are two sure payoffs $x_{0},y_{0}\in\mathcal{X}$, i.e., $x_{0},y_{0}$
are constant, such that $x_{0}<y_{0}$ but the distribution of $y_{0}$
does not dominate the distribution function of $x_{0}$ in integral
stochastic ordering. 
\end{remark}

\begin{example}
\label{exa:Thm5}Assume $\mathbb{F}=\mathbb{F}_{SSD}$. Consider the
robust rank-dependent expected utility maximization problem in Example
\ref{exa:RDU} with solution $X_{\mu_{1},s}^{\ast}$ defined in 
(\ref{eq:solution_RDU}). Let $\mathbb{P}^{\ast}=\mathbb{P}^{\mu_{1},s}$
and $\ell^{\ast}=\ell^{\mathbb{P}^{\ast}}$. With the help of the
explicit expressions of $F_{X_{\mu_{1},s}^{\ast}}^{\mathbb{P}^{\ast}}$
and $F_{\ell^{\ast}}^{\mathbb{P}^{\ast}}$ from the proofs of Examples
\ref{exa:RDU} and~\ref{exa:EUT}, it is easy to see that (b) in Theorem
\ref{thm:ratio_SSD} holds if $\gamma<0$.%

Let us start from viii) in Theorem~\ref{thm:ratio_SSD}. Equation (\ref{eq:solution_RDU})
implies that the optimal solution $X_{\mu_{1},s}^{\ast}$ is a non-decreasing
function of $\ell^{\ast}$. Hence, iii) in Theorem~\ref{thm:ratio_SSD}
is satisfied. As shown by Example~\ref{exa:EUT}, $X_{\mu_{1},s}^{\ast}$
also solves an expected utility maximization problem for the utility
in (\ref{eq:u}), which illustrates v) in Theorem~\ref{thm:ratio_SSD}.
One can easily verify that ii) in Theorem~\ref{thm:ratio_SSD} is
respected by $X_{\mu_{1},s}^{\ast}$ . Hence, Theorem~\ref{thm:robust-cost-payoff}
implies that $X_{\mu_{1},s}^{\ast}$ solves the robust $\mathbb{F}_{SSD}-$cost-efficiency
problem as stated in iv) in Theorem~\ref{thm:ratio_SSD}. 
\end{example}

The preferences in Theorem~\ref{thm:ratio_SSD} for case (b) do not
need to be law-invariant or increasing, i.e., $X\leq Y$ does not
need to imply $W_{\mathbb{P}}(X)\leq W_{\mathbb{P}}(Y)$. We provide
a simple example of such a preference.
\begin{example}
Define\emph{ }$X^{\ast}=f(\ell^{\ast})$ for some $f\in\mathbb{F}$.
Let $x_{0}=E_{\mathbb{Q}}[X^{\ast}]$. For $\mathbb{P}\in\mathcal{P}$,
define 
\[
W_{\mathbb{P}}(X)=\begin{cases}
1 & \text{if}\ X=X^{\ast},\,\,\,\mathbb{P}-\text{a.s.}\\
0 & \text{if}\ \text{ otherwise},
\end{cases}
\]
which is trivially $\mathbb{F}-$family consistent. An agent with
such a preference only likes $X^{\ast}$ and neglects everything else.
She is not law-invariant and does not prefer more to less. Someone
interested only in the market portfolio or in the risk-free bond might
have such a preference. The solution to the robust maximization problem
ix) is $X^{\ast}$, which is cost-efficient because $X^{\ast}$ is
non-decreasing in $\ell^{\ast}$. A utility function for
v) can be constructed as in (\ref{eq:u-1}). 
\end{example}

\section{\label{sec:Summary}Final Remarks}

In this paper we assume that the agent has Knightian uncertainty. She is
unsure about the precise physical measure describing the financial
market and knows only that the true physical measure lies within a
set $\mathcal{P}$ of probability measures. Given this ambiguity,
it is no longer possible to target a payoff with a given probability
distribution function. In particular, the close relation between payoffs
that are the cheapest possible in reaching a target distribution function
and the optimality thereof under law-invariant increasing preferences
(\cite{dybvigJoB,dybvigRFS}, \cite{sharpe2000distribution}, \cite{goldstein2008choosing},
\cite{bernard2015rationalizing}) is a priori lost, as there is no
consensus regarding what probability distribution to adopt. For this
reason, we introduce the notion of a \textit{robust cost-efficient payoff}.

For a given distribution function $F_{0}$, the robust cost-efficiency
problem aims at finding the cheapest payoff whose distribution function
dominates $F_{0}$ under all possible physical measures in some integral
stochastic ordering. We solve this problem under some conditions (namely,
where there exists a least favorable measure $\mathbb{P}^{\ast}$
and the integral stochastic ordering $\preceq_{\mathbb{F}}$ is cost-consistent).
The solution is identical to the solution to the cost-efficiency problem
without model ambiguity under the physical measure $\mathbb{P}^{\ast}$
and given in closed-form. We are thus able to reduce the problem formulated
in a robust setting to a problem formulated in a standard setting
without model ambiguity.

Finally, we show that this notion of robust cost-efficiency plays
a key role in optimal robust portfolio selection and that a very general
class of robust portfolio selection problems (possibly in a non-expected
utility setting) can be reduced to the maxmin expected utility setting
of \cite{gilboa1989maxmin} for a well-chosen concave utility function.

For this to hold, we make a relatively minor assumption on the family
of preferences, i.e., that it is family consistent: if the measure
$\mathbb{P}^{\ast}$ is the most pessimistic view of a payoff $Y$,
then the preference under that measure is the lowest as well. To the
best of our knowledge, family consistency is new to the literature,
and we provide several examples in the context of expected utility
theory, Yaari's dual theory and rank-dependent utility theory.

We assume a static setting in which intermediate trading is not possible. Whilst allowing for dynamic rebalancing may make it possible to achieve higher levels for the objective at hand (e.g., robust expected utility), this possibility is only clear when there are no transaction costs, which is not realistic. In practice, transaction costs usually contain a fixed part, and hence dynamic trading can only occur a finite number of times since otherwise bankruptcy occurs. The study of optimal investments in the presence of fixed costs is not yet very well understood. Recently, \cite{Belak2022} and \cite{Bayraktar2022} provide optimal strategies in a Black-Scholes market without ambiguity and assuming expected utility. By contrast our static setting makes it possible to deal with ambiguity and to address fairly general objectives. 


\appendix
\begin{center}
\textbf{\LARGE{}Appendix}{\LARGE{} }{\LARGE\par}
\par\end{center}

\section{\label{sec:Proof-of-TSD2} Proof for Example~\ref{lem:TSD}}
\begin{lemma}
\label{Ltsd} Let $F$ and $G$ be two cdfs. It holds that \emph{ $G\preceq_{\mathbb{F}_{TSD}}F$}
if and only if 
\[
\int_{-\infty}^{\eta}\int_{-\infty}^{\xi}F(x)dxd\xi\leq\int_{-\infty}^{\eta}\int_{-\infty}^{\xi}G(x)dxd\xi,\,\,\,\eta\in\mathbb{R}.
\]
\end{lemma}

\begin{proof}
See Theorem 2.2 of \cite{gotoh2000third}. 
\end{proof}
We are now ready to construct the counterexample stated in Example
\ref{lem:TSD}. 
\begin{proof}
Apply the chain rule to show that $\mathbb{F}_{TSD}$ is composition-consistent.
Next, we construct two distribution functions such that one dominates
the other in TSD but is cheaper.\emph{}\\
\emph{Step 1: define some market setting as in Section~\ref{subsec:Example:-Uncertainty-about-drift}:}
let $\mu_{1}=0.01$, $r=0$, $T=1$ and $s=0.1$. Then, $\frac{\mu_{1}-r}{s^{2}}=1$.
Choose $S_{0}$ such that it holds that $\ell^{\ast}:=\ell^{\mathbb{P}^{\mu_{1}}}=\ensuremath{h^{\mu_{1},s}(S_{T})}=S_{T}$,
i.e., $\log(S_{0})=-0.0025$. Under $\mathbb{P}^{\ast}:=\mathbb{P}^{\mu_{1}}$
the stock is log-normally distributed with parameters $(\mu_{1}-\frac{s^{2}}{2}+\log(S_{0}))=0.0025$
and $s^{2}=0.01$. Under $\mathbb{Q}$, the stock is also log-normally
distributed with parameters $(r-\frac{s^{2}}{2}+\log(S_{0}))=-0.0075$
and $s^{2}=0.01$. $\mathbb{P}^{\ast}$ is a least favorable measure
with respect to the set $\mathbb{F}_{FSD}$, and so hence also is $\mathbb{F}_{TSD}$
because $\mathbb{F}_{TSD}\subset\mathbb{F}_{FSD}$.\emph{}\\
\emph{Step 2: define two distribution functions:} Let 
\[
F(x)=\begin{cases}
0 & ,x<0\\
x & ,0\leq x<1\\
1 & ,x\geq1
\end{cases}
\]
and, for $p_{0}\in(0,1)$, let 
\[
G(x)=\begin{cases}
0 & ,x<0\\
p_{0} & ,0\leq x<1\\
1 & ,x\geq1.
\end{cases}
\]
$F$ is the uniform distribution function and $G$ jumps at zero and
at one. It follows that $F^{-1}(p)=p$ and that 
\[
G^{-1}(p)=\begin{cases}
0 & ,p\in(0,p_{0}]\\
1 & ,p>p_{0}.
\end{cases}
\]
\emph{Step 3: show that $F$ dominates $G$ in TSD: }It holds for
$\eta\in(0,1)$ that 
\[
\int_{-\infty}^{\eta}\int_{-\infty}^{\xi}F(x)dxd\xi=\int_{0}^{\eta}\int_{0}^{\xi}xdxd\xi=\frac{1}{6}\eta^{3}.
\]
and that 
\[
\int_{-\infty}^{\eta}\int_{-\infty}^{\xi}G(x)dxd\xi=\int_{0}^{\eta}\int_{0}^{\xi}p_{0}dxd\xi=\frac{1}{2}p_{0}\eta^{2}.
\]
Hence, if $\frac{1}{6}\eta^{3}\leq\frac{1}{2}p_{0}\eta^{2}$ or, equivalently,
$p_{0}\geq\frac{1}{3}$, it follows that $G\preceq_{\mathbb{F}_{TSD}}F$.\emph{}\\
\emph{Step 4: compute the lowest cost of both distribution functions:
}The cost-efficient payoff for $F$ is 
\[
X_{F}=F^{-1}\left(F_{\ell^{\ast}}^{\mathbb{P}^{\ast}}(\ell^{\ast})\right)=F_{S_{T}}^{\mathbb{P}^{\ast}}(S_{T}).
\]
The lowest price of $F$ can be computed numerically: 
\[
E_{\mathbb{Q}}[X_{F}]=\int_{0}^{\infty}F_{S_{T}}^{\mathbb{P}^{\ast}}(s)f_{S_{T}}^{r}(s)ds=0.472.
\]
The cost-efficient payoff for $G$ is 
\[
X_{G}=G^{-1}\left(F_{\ell^{\ast}}^{\mathbb{P}^{\ast}}(\ell^{\ast})\right)=\begin{cases}
1 & ,S_{T}>\left[F_{S_{T}}^{\mathbb{P}^{\ast}}\right]^{-1}(p_{0})\\
0 & ,\text{otherwise}
\end{cases}.
\]
Its price is 
\[
E_{\mathbb{Q}}[X_{G}]=\int_{\left[F_{S_{T}}^{\mathbb{P}^{\ast}}\right]^{-1}(p_{0})}^{\infty}f_{S_{T}}^{r}(s)ds=1-F_{S_{T}}^{\mathbb{Q}}\left(\left[F_{S_{T}}^{\mathbb{P}^{\ast}}\right]^{-1}(p_{0})\right).
\]
Under $\mathbb{P}^{\ast}$, $X_{F}$ is uniform distributed and $X_{G}$
is a digital option. If $p_{0}=\frac{1}{3}$, the lowest price for
$G$ is $0.63$, which is greater than the lowest price to be paid
for $F$. But in this case $G\preceq_{\mathbb{F}_{TSD}}F$: hence,
TSD is not cost-consistent. 
\end{proof}

\section{Auxiliary results}
\begin{lemma}
\label{lem:N_cost_eff_P}Fix $\mathbb{P}\in\mathcal{P}$. Let $\ell:=\frac{d\mathbb{P}}{d\mathbb{Q}}$
be the Radon--Nikodym derivative of $\mathbb{P}$ with respect to
$\mathbb{Q}$. Assume that under $\mathbb{P}$ $\ell$ is continuously distributed and that $1/\ell$ has finite variance. There is a $\mathbb{P}-$a.s. unique optimizer to the standard cost-efficiency
problem under the probability measure \emph{$\mathbb{P}$} given by
\[
X^{\ast}=F_{0}^{-1}\left(F_{\ell^{\mathbb{P}}}^{\mathbb{P}}\left(\ell\right)\right).
\]
$X^{\ast}$ is left-continuous and non-decreasing $\mathbb{P}-$a.s. 
\end{lemma}

\begin{proof}
Let $\xi=\frac{e^{-rT}}{\ell}$. Then, $1-F_{\xi}^{\mathbb{P}}(\xi)=F_{\ell^{\mathbb{P}}}^{\mathbb{P}}\left(\ell\right).$
This claim follows both from \cite{dybvigJoB,dybvigRFS} and from Corollary 2 in
\cite{bernard2013explicit}. See also \citet[Proposition 2.7]{schied2004neyman} for the importance of the  continuity assumption of $\ell$ in obtaining the uniqueness.

\end{proof}
\begin{lemma}
\label{lem:cost_eff_non_decreasing}Fix $\mathbb{P}\in\mathcal{P}$.
Let $\ell:=\frac{d\mathbb{P}}{d\mathbb{Q}}$. Assume that $\ell$ is continuously
distributed under $\mathbb{P}$. A payoff $X\in\mathcal{A}_{F_{0}}^{\mathbb{P}}$
is \emph{$\mathbb{P}-$}cost-efficient if and only if it is non-decreasing
in $\ell$, $\mathbb{P}-$almost surely. 
\end{lemma}

\begin{proof}
See \citet[Corollary 2 and Proposition 2]{bernard2013explicit}.
\end{proof}

In the following two lemmas we show that the solution of the
single or robust maximization problem is cost-efficient if it is unique. 
\begin{lemma}
\label{cor:unique_P*}Let $\mathbb{P}\in\mathcal{P}$
with corresponding likelihood ratio $\ell^{\mathbb{P}}$. Assume that
$\ell^{\mathbb{P}}$ is continuously distributed and that $W_{\mathbb{P}}$
is $\mathbb{P}-$law invariant and $\tilde{X}$ is a $\mathbb{P}-$a.s.
unique solution to the maximization problem (\ref{eq:singleMaximization})
under $\mathbb{P}$. Then, $\tilde{X}$ is $\mathbb{P}-$cost-efficient. 
\end{lemma}

\begin{proof}
Let 
\[
X^{\ast}=\left[F_{\tilde{X}}^{\mathbb{P}}\right]^{-1}\left(F_{\ell^{\mathbb{P}}}^{\mathbb{P}}(\ell^{\mathbb{P}})\right).
\]
Then $X^{\ast}$ solves the standard cost-efficiency problem for $F_{\tilde{X}}^{\mathbb{P}}$
and thus $E_{\mathbb{Q}}[X^{\ast}]\leq E_{\mathbb{Q}}[\tilde{X}]$
and $F_{\tilde{X}}^{\mathbb{P}}=F_{X^{\ast}}^{\mathbb{P}}$; hence,
by the law invariance of $\left(W_{\mathbb{P}}\right)_{\mathbb{P}\in\mathcal{P}}$,
it holds that $X^{\ast}\in\mathcal{Y}_{W_{\mathbb{P}}}^{x_{0}}$.
It follows by law invariance that 
\[
\underset{X\in\mathcal{Y}_{W_{\mathbb{P}}}^{x_{0}}}{\max}\,W_{\mathbb{P}}(X)=W_{\mathbb{P}}(\tilde{X})=W_{\mathbb{P}}(X^{\ast}).
\]
As $\tilde{X}$ is the unique solution, it must hold that $\tilde{X}=X^{\ast}$,
$\P-$a.s., and thus $\tilde{X}$ is cost-efficient. 
\end{proof}
\begin{lemma}
\label{cor:unique_robust}Assume $\mathbb{F}=\mathbb{F}_{FSD}$. Given
Assumptions \ref{assu:setF}, \ref{assu:least_fav} and \ref{assu:l*},
assume that the robust maximization problem (\ref{eq:robust maximization problem})
has a unique solution $\tilde{X}$ and that $\left(W_{\mathbb{P}}\right)_{\mathbb{P}\in\mathcal{P}}$
is law invariant and $\mathbb{F}_{FSD}-$family consistent on $\mathcal{Y}_{(W_{\mathbb{P}})_{\mathbb{P}\in\mathcal{P}}}^{x_{0}}$
with respect to $\mathbb{P}^{\ast}$. Then, $\tilde{X}$ is $\mathbb{P}^{\ast}-$cost-efficient. 
\end{lemma}

\begin{proof}
The proof is similar to the one for Lemma \ref{cor:unique_P*}: let 
\[
X^{\ast}=\left[F_{\tilde{X}}^{\mathbb{P}^{\ast}}\right]^{-1}\left(F_{\ell^{\ast}}^{\mathbb{P}^{\ast}}(\ell^{\ast})\right).
\]
It holds by law invariance that 
\[
\underset{X\in\mathcal{Y}_{(W_{\mathbb{P}})_{\mathbb{P}\in\mathcal{P}}}^{x_{0}}}{\max}\,\underset{\mathbb{P}\in\mathcal{P}}{\inf}\,W_{\mathbb{P}}(X)=\underset{\mathbb{P}\in\mathcal{P}}{\inf}\,W_{\mathbb{P}}(\tilde{X})=\underset{\mathbb{P}\in\mathcal{P}}{\inf}\,W_{\mathbb{P}}(X^{\ast}).
\]
\end{proof}

\bibliography{biblio}
{} \bibliographystyle{plainnat}

\end{document}